\documentclass[%
 reprint,
 amsmath,amssymb,
 aps,
]{revtex4-2}

\usepackage{graphicx}
\usepackage{dcolumn}
\usepackage{bm}
\usepackage{enumerate}

\usepackage{amsthm,amsmath,amssymb}
\usepackage{mathrsfs}
\usepackage{bm}
\usepackage{lineno}
\newtheorem{definition}{Definition}
\newtheorem{theorem}{Theorem}
\newtheorem{lemma}{Lemma}

\begin{document}

\preprint{APS/123-QED}

\title{Excited state fluid mechanics and mathematical principles of separation and transition}

\author{Peng Yue}
\email[Emial: ]{pengyu.yue@outlook.com}
\affiliation{University of Electronic Science and technology of China, Chengdu, China}
\affiliation{China Aerodynamics Research and Development Center, Mianyang, China}
\author{Jinping Xiao}
\affiliation{University of Electronic Science and technology of China, Chengdu, China}
\affiliation{China Aerodynamics Research and Development Center, Mianyang, China}
\author{Ke Xu}
\affiliation{University of Electronic Science and technology of China, Chengdu, China}
\author{Ming Li}
\affiliation{China Aerodynamics Research and Development Center, Mianyang, China}
\author{Feng Jiang}
\affiliation{Northwestern Polytechnical University, Xi'an, China}
\author{Yiyu Lu}
\affiliation{University of Electronic Science and technology of China, Chengdu, China}
\author{Dewei Peng}
\affiliation{China Aerodynamics Research and Development Center, Mianyang, China}



\begin{abstract}
Transition and separation are difficult but important problems in the field of fluid mechanics. Hitherto, separation and transition problems have not been described accurately in mathematical terms, leading to design errors and prediction problems in fluid machine engineering. The nonlinear uncertainty involved in separation and transition makes it difficult to accurately analyze these phenomena using experimental methods. Thus, new ideas and methods are required for the mathematical prediction of fluid separation and transition. In this article, after an axiomatic treatment of fluid mechanics, the concept of an excited state is derived by generating a fluctuation velocity, and it is revealed that fluid separation and transition are special forms of this excited state. This allows us to clarify the state conditions of fluid separation and transition. Mathematical analysis of the Navier--Stokes equations leads to a general excited state theorem suitable for flowfields. Finally, the conditions of separation and transition are derived, and the corresponding  general laws are established. The results presented in this article provide a foundation for future research on the mechanism of turbulence and the solution of engineering problems.
\end{abstract}
\maketitle


\section{\label{sec:level1}Intrdoction}

Over the past 100 years, there have been considerable development in fluid mechanics. Advances in aeronautical engineering, ocean engineering, petroleum engineering, and other disciplines all depend on the continuous innovation of fluid mechanics engineering investigations. In recent years, the rapid rise of advanced experimental technology and computational fluid dynamics (CFD) have clarified the state of flowfield motion through digital methods, and have revealed more meaningful scientific problems, such as separation and transition. \par
The earliest study of separation and transition phenomena began in 1883. Renault discovered two different fluid forms through experiments on circular pipe flow, and defined the relevant dimensionless parameters.\cite{ Hall1981, Viviand1987} However, there is still no general method to determine the mechanism of the two flow states. Unsteady flow separation in rotationally augmented flowfields plays a significant role in a variety of fundamental flows.\cite{Melius2018} The aerodynamic performance of lifting surfaces operating under low-Reynolds-number conditions is impaired by separation,\cite{Melius2016} and the transition from laminar flow to turbulent flow is of great practical interest.\cite{Yang2019} To solve these problems, scholars have conducted a range of meaningful experimental and numerical studies. In addition, the separation induced by the shock wave/boundary layer interaction is obvious. As a common phenomenon, a large number of scholars have performed the reaserches in recent years. Huang et al.\cite{Huang2020,Yan2020,Du2021} studied and pointed out that micro vortex generators have been widely employed in the internal and external flowfields to suppress the separation caused by the interaction between shock wave and boundary layer, and to develop quantitative measurement and evaluation methods. As one of the most significant researches in recent years, these work not only systematically summarize and expound the major breakthroughs in related fields, which plays an important role in the further development of precision fluid machinery, but also describe the geometric and physical characteristics of supersonic/hypersonic flow and application of flow control techniques.\par
Experimental methods are relatively accurate means of describing a physical process. Recently, separation and transition have been widely investigated through experiments. Croci et al.\cite{Croci2019} revealed the emergence of two large laminar boundary layer separations downstream of the Venturi throat, in addition to low pressure zones that may induce both degassing or cavitation features. Through the use of time-resolved particle image velocimetry, Melius et al.\cite{Melius2018} examined vorticity accumulation and vortex shedding during unsteady separation over a three-dimensional airfoil. Chandra et al.\cite{Chandra2019} investigated the laminar--turbulent transition in the flow of Newtonian and viscoelastic fluids in soft-walled microtubes of diameter 400 $\mu m$ using micro-particle image velocimetry, and found that the fluid and wall elasticities are combined to trigger a transition at Reynolds numbers as low as 100 in the flow of polymer solutions through deformable tubes. Miro et al.\cite{Miro2019} quantified the influence of transition, diffusion, collision, equilibrium, and chemical-kinetics modeling on the stability characteristics of a flow, and estimated the transition-onset location of canonical boundary layers. They found that the boundary-layer height calculation is paramount to the simulation of the development of second-mode instabilities. Istvan et al.\cite{Istvan2018} experimentally studied the effects of the freestream turbulence intensity on the mean topology and transition characteristics of laminar separation bubbles forming over the suction side of a NACA 0018 airfoil. Wei et al.\cite{Wei2019} investigated the boundary layer transition and separation on an oscillating S809 airfoil using pressure signatures captured in wind tunnel tests. This provided a better understanding of the unsteady aerodynamic characteristics of the airfoil, which shows important roles in wind turbine blade design. Grossman et al.\cite{Grossman2018} observed the resulting shock wave/boundary layer interaction at the tunnel wall where an oblique shock wave is generated in a Mach-2 flow at a deflection angle of 12$^{\circ}$. Kuhnen et al.\cite{Kuhnen2018} found that the turbulent strength quickly decreases when a rotor is used in a circular tube to enhance the turbulence level, with severe turbulence flowing through the downstream rotor. They also observed that the turbulent strength quickly decreases in a laminar flow. These experiments showed that nonlinear effect can be counteracted by other nonlinear effect.
Although experimental methods can describe the physical process accurately, subtle differences between the experimental environment and the actual environment may affect the nonlinear initial values and introduce significant errors into the final results. In addition, the high economic cost of experimental apparatus makes such researches difficult. For these reasons, the development and promotion of theories are restricted. \par
Numerical methods aim to simulate the physical scene through advanced computing equipments, thereby recreating the state and dynamical properties of the fluid flow. This can greatly reduce the research cost and allow the ideal physical models to be analyzed in detail. As the research results are intuitive and vivid, they are of increasing interest to scholars. The same goes for transition and separation issues. Sengupta et al.\cite{Sengupta2020} solved the Navier--Stokes equations to investigate the individual and cumulative effect of forced frequency oscillations and freestream turbulence on the separation-induced transition caused by an adverse pressure gradient on a flat plate geometry. Salimipour\cite{Salimipour2019} studied a two-dimensional numerical simulation of the incompressible transitional flow around the NACA 0012 and Eppler 387 airfoils. This study reported better predictions of the separation bubbles than previous results, especially for long separation bubbles. Although the Reynolds time-averaged equation is widely used in engineering, it is not conducive to the study of physical laws because of its relatively poor accuracy. Therefore, large eddy simulation (LES) and direct numerical simulation (DNS) have been widely developed. Ni et al.\cite{Ni2019} studied turbulent boundary-layer separation from a backward-facing rounded ramp with active wall actuation control through an implicit LES approach, while Hosseinverdi et al.\cite{Hosseinverdi2018} carried out highly resolved DNS in which very-low-amplitude isotropic freestream turbulence fluctuations were introduced at the inflow boundary of the computational domain. Jiang et al.\cite{Jiang2020} also used a DNS method to examine the flow separation around a square cylinder for Reynolds numbers from 10-400. \par
In the field of CFD, the biggest disadvantage of LES and DNS is their high computational cost. Duraisamy et al.\cite{Duraisamy2018} exploited foundational knowledge in turbulence modeling and physical constraints to derive useful predictive models using data-driven approaches. To achieve the goals of engineering computing more quickly, a large number of verifications and improvements have been made to models and methods. Hussin et al.\cite{Hussin2018} discretized the compressible Navier--Stokes equations using a finite volume method and solved them with a semi-implicit pressure linking algorithm on unstructured grids. Wang et al.\cite{Wang2020} found that the shear-stress transport (SST)-$\gamma$ model captures an laminar separation bubble near the leading edge of the airfoil and shows significant advantages over traditional ``fully turbulent'' models for the prediction of static stall. Bernardos et al.\cite{Bernardos2019} developed algebraic transitional extensions for the accurate computation of laminar separation bubbles with $k-\omega$ models. Xu et al.\cite{Xu2019} proposed a two-equation transport transition model based on the analysis of linear stability theory, and built a CFD-compatible transition model. Mishra et al.\cite{Mishra2019} found that the nonlinear $k-\omega$ SST and $k-k_L-\omega$ transition models provide comparable predictions of the lift and drag coefficients. Bernardos et al.\cite{Bernardos2019} investigated the performance of a recently developed laminar separation transition triggering approach that controls the transitional turbulence source terms. Hu et al.\cite{Hu2019} studied laminar juncture flow numerically to investigate the outermost saddle point of the SS/SA and the topological transition.\par
These studies have played important roles in promoting the development of fluid mechanics. From their results, we can summarize some important and generally accepted views. Separation and transition are closely related to the pressure gradient. Separation refers to reflux, and transition refers to an irregular flow. From the perspective of phenomenological physics, the randomness and uncertainty of nonlinear effect make it difficult to summarize the general laws of experimental and numerical methods, so fundamentally solving the separation and transition problems is challenging. This has had a negative impact on further innovation and optimization in the most important engineering fields, such as flow control,\cite{Shahrabi2019, Li2019} microflow control,\cite{Beebe2015} aircraft design under extreme conditions,\cite{Kwiek2019} low-altitude flight stability,\cite{Tang2019, Zhou2019} superfluids,\cite{Johnstone2019} fluid materials,\cite{Huang2019} flight test applications,\cite{Wang2019} and innovative design.\cite{Kwiek2019} An universal law of separation and transition would enable significant innovations in other engineering fields. Thus, its importance cannot be ignored. \par
In this article, separation and transition are studied through rational mechanics methods with the aim of producing better predictions on the separation and transition positions. According to Euler's description, the velocity of a fluid at a certain point in space can be considered as the instantaneous velocity, and the instantaneous velocity minus the averaged velocity is the fluctuation velocity. The concept of the excited state is defined by the generation of fluctuation velocity. Transition and separation are defined through the phenomenological characteristics of this excited state and the fluctuation velocity. Using the idea of superposition states in quantum mechanics, the generation of the excited state is then analyzed mathematically. Finally, the necessary and sufficient conditions for transition and separation are studied, and the results are compared with experimental counterparts for the separation and transition in the upper part of an airfoil. The predictions are found to be accurate to within $2\%$.

\section{\label{sec:level2}General mathematical model of flowfield}
In the field of fluid mechanics, the Navier--Stokes equations describe the flowfield. Assuming that the fluid domain is $\Omega \subseteq \mathbb{R}^3$ and the velocity vector field is in the differential manifold constructed by $M=\left(\Omega,\mathscr{S}\right)$, whose cotangent vector field represents its velocity field $u_i \in \mathscr{T}_M$. The $m$-th submanifold in $M$ is $M_m$. As the bounded space $\Omega$ is homeomorphic with the smooth submanifold $M_m$ with edges \cite{Teleman2019}, we could find a mapping $\varphi_m$ satisfying the following conditions:
\begin{equation} 
\label{1}
\varphi_m: x_j \mapsto u_i, \ x_j \in \Omega, \ u_i \in \mathscr{T}_{M_m}.
\end{equation} 

In this way, a complex flow can be represented by embedding different submanifolds and different states can be considered as different submanifolds. In this article, to simplify the mathematical expressions and clarify the physical meaning, the mathematical laws of submanifolds are not discussed; they will be discussed in subsequent studies. Therefore, for any velocity fields of isotropic fluid medium, the spatial and temporal evolution of the physical field should satisfy the Navier--Stokes equations \cite{Buckmaster2019}.
\begin{lemma}[Navier--Stokes equations] For $\Omega \subseteq \mathbb{R}^3$, $u_i=\varphi_m\left(x_j\right) \in \mathscr{T}_M$, and $x_j \in \Omega$, the velocity field $u_i$ satisfies the continuity equation
\begin{equation}
\label{2}
\frac{\partial \rho}{\partial t}+\left(\rho u_{i}^{\bullet}\right)^{, i}=0 
\end{equation}
and the momentum equation for a viscous medium
\begin{equation}
\label{3}
\rho \frac{\partial u_{i}}{\partial t}+\rho u_j u_{i,\bullet}^{\bullet, j}=\sigma _{ij,\bullet}^{\bullet\bullet,j}+ f_i,
\end{equation}
where $\rho$ is the fluid density, and $f_i$ is the body force. $\sigma _{ij}$ is the stress state tensor for isomorphic Newtonian fluids, which is expressed as follows:
\begin{equation}
\label{4}
\sigma _{ij}=-p\delta_{ij}+\psi \mu u_{k,\bullet}^{\bullet,k}\delta_{ij}+\mu\left(u_{j,i}+u_{i,j}\right),
\end{equation}
here $p$ is the pressure, $\delta _{ij}$ is the Kronecker symbol, and $\mu$ is coefficient of dynamic viscosity.
\end{lemma}

In the construction process of the Reynolds equation,  the fluctuation velocity  \cite{Yan2019} is introduced to reflect the randomness and uncertainty of turbulence. Similarly, it is also used here to reveal the mechanism of separation and transition phenomenologically.

\begin{definition}[Fluctuation velocity]
Let $\Omega \in \mathbb{R}^3$ be the flowfield space, and the velocity vector field is in a differential manifold constructed by $M=\left(\Omega,\mathscr{S}\right)$, whose cotangent vector field represents its velocity field $u_i \in \mathscr{T}_M$. The velocity field can be expressed as:
\begin{equation}
\label{5}
u_i={u'}_i+\bar{u}_i,
\end{equation}
where ${u'}_i \in \mathscr{T}_M$ is the fluctuation velocity field and $\bar {u}_i \in \mathscr{T}_M$ is the average velocity field.
\end{definition}

Note that the fluctuation velocity may introduce nonlinearity and bifurcations. As we define and embed different submanifolds, they should still satisfy the mapping conditions of Eq.~(\ref{1}). This is an important problem, and will be the focus of future research.\par

If the above relationship is combined with the Navier--Stokes equations, the Reynolds equation \cite{Du2019} can be obtained. In fluid dynamics, turbulence is a chaotic state of motion \cite{Graham2021}. Although difficult to explain, it is a very common phenomenon, causing many problems in practical engineering applications.\par
Note that the fluctuation velocity defined here is not necessarily irregular and disordered. Apparently, this definition is only a simple superposition principle of vector fields, but we can try to define transition and separation within this framework.\par
The key features of transition \cite{Majumdar2020} are randomness and irregularity, so the fluctuation velocity induced by different points should be defined differently. The transition position is a set of some points, at which the laminar flow converts into turbulent flow, and where the flow moves from being stable to unstable. From the perspective of phenomenological physics, transition refers to the generation of irregular and random flow mechanisms. According to the second law of thermodynamics, the system spontaneously moves towards chaos. Therefore, it is especially significant to define the first point or region where this irregular change occurs. The first point that makes the fluid chaotic must have the same velocity as other points, but its fluctuation velocity will be different.\par

\begin{definition}[Transition]
Let $U\left(x_h,\epsilon\right)$ be a neighborhood of point $x_h \in \Omega$. A laminar fluid with velocity $u_i \in \mathscr{T}_{M}$ flows through the neighborhood, and excites the fluctuation velocity ${u'}_j \in \mathscr{T}_M$ through various disturbances. Let $\forall x_k \in U\left(x_h,\epsilon\right)$, $\exists x_s \in U\left(x_h,\epsilon\right)$, and $x_s \neq x_k$. If they satisfy
\begin{equation}
\label{6}
u_i\left(x_k\right)=u_i\left(x_s\right),\ {u'}_j\left(x_k\right) \neq {u'}_j\left(x_s\right), \ {u'}_j\left(x_h\right)=0,
\end{equation}
then there is a transition in the flowfield at point $x_h$. The point $x_h$ is called the transition point, and the set of transition points is called the transition position. 
\end{definition}

Separation is another key physical feature in fluid mechanics, representing a huge loss of flow energy \cite{Williams2021}. It cannot be ignored in the process of fluid machinery design. For separation, it is not important whether the fluctuation velocity is regular or not. In phenomenological physics, it is more important to determine whether the fluctuation velocity in the neighborhood is opposite to the original velocity distribution or not near the wall. Reason for the separation phenomenon is backflow, that is, a reverse velocity is equal to or greater than the original one. Thus, we define separation as follows.
\begin{definition}[Separation]
Let $U\left(x_h,\epsilon\right)$ be a neighborhood of point $x_h \in \Omega$. A fluid with velocity $u_i \in \mathscr{T}_{M}$ flows through the neighborhood, and excites the regular distribution of fluctuation velocity ${u'}_j \in \mathscr{T}_M$ through various disturbances. If $\exists x_k \in U \left(x_h,\epsilon\right)$ such that
\begin{equation}
\label{7}
u_i\left(x_k\right){u'}_j\left(x_k\right) \leq 0, \ \left|u_i\left(x_k\right)\right|\leq \left|{u'}_j \left(x_k\right)\right|,
\end{equation}
$$\ {u'}_j\left(x_h\right) = 0,$$
then there is a separation in the flowfield at point $x_h$. The point $x_h$ is named the separation point, and the set of these points is the separation position. If the distribution of fluctuation velocity is irregular, the irregular distribution of fluctuation velocity will make the fluid attach to the surface, which is called separation effect, but not separation. 
\end{definition}

We now find that the emergence of the fluctuation velocity is one of the key factors in solving the transition and separation problems. The fluctuation velocity is a physical variable that cannot be described directly, so we first study the characteristics of the fluctuation velocity and related physical variables at the transition and separation positions.\par

\section{\label{sec:level3}Mathematical and physical conditions for generation of fluctuation velocity}

To accurately depict the flowfield variation characteristics from the mathematical and physical perspectives, it is assumed that a new fluctuation velocity field is generated in the flowfield. Before the fluctuation velocity is generated, the value is zero. At the next moment, it is not equal to zero, which means that the derivative of the fluctuation velocity with respect to time is not equal to zero in the generation process of fluctuation velocity field. Therefore, we could draw the following conclusions.

\begin{theorem}[Conditions for generating fluctuation velocity]
Let point ${x_k} \in \Omega$ have a neighborhood $U\left(x_k,\varepsilon\right)$. For $x_j \in \{x_i\vert x_i \in U\left(x_k,\varepsilon\right), \left| x_k -x_i \right| <\varepsilon, \forall\varepsilon>0 \}$, the conditions of generating a fluctuation velocity can be expressed as
\begin{equation}
\label{8}
{u'}_j=0,
\end{equation}
and the temporal derivative satisfies as follows:
\begin{equation}
\label{9}
\frac{d{u}'_j}{dt}=\frac{\partial {u}'_j}{\partial t}+ {u}_k {u'}^{\bullet, k}_{j,\bullet} \neq 0.
\end{equation}
As $\varepsilon \rightarrow 0$, $x_k$ gives the coordinates of key points such as the transition and separation points. 
\end{theorem}

Eqs.~(\ref{8}) and (\ref{9}) could be regarded as the necessary conditions for separation and transition. When a fluid particle is at the key position, although the fluctuation velocity has not been generated, the fluctuation acceleration exists and induces the appearance of fluctuation velocity at the next moment. At this time, the flowfield begins to separate or transite.\par

Although we have presented particularly important conditions for the occurrence (and even the end) of the fluctuation velocity as a random and uncertain physical variable, it is difficult to evaluate the law governing its development. Therefore, we introduce the superposition principle of quantum mechanics to obtain the general dynamic law of fluctuation velocity in fluid mechanics.\par

\section{\label{sec:level4}Spatial and temporal evolution law of excited and unexcited states}
In this section, we introduce the concept of the superposition state (Fig.~\ref{p1}) in quantum mechanics \cite{Chiribella_2020}. We assume that any uncertain state in an event exists in different parallel spaces. These parallel spaces will eventually collapse into a single state for some reason (according to current researches in the field of physics, this reason can be regarded as subjective in the field of quantum mechanics, but it is generally regarded as objective in the field of classical mechanics). By using the idea of the superposition state, the reason for the existence of separation and transition can be determined. Firstly, the spatial and temporal evolution law of fluctuation velocity is studied.\par

\begin{figure}[ht]
\includegraphics[scale=0.35,angle=0]{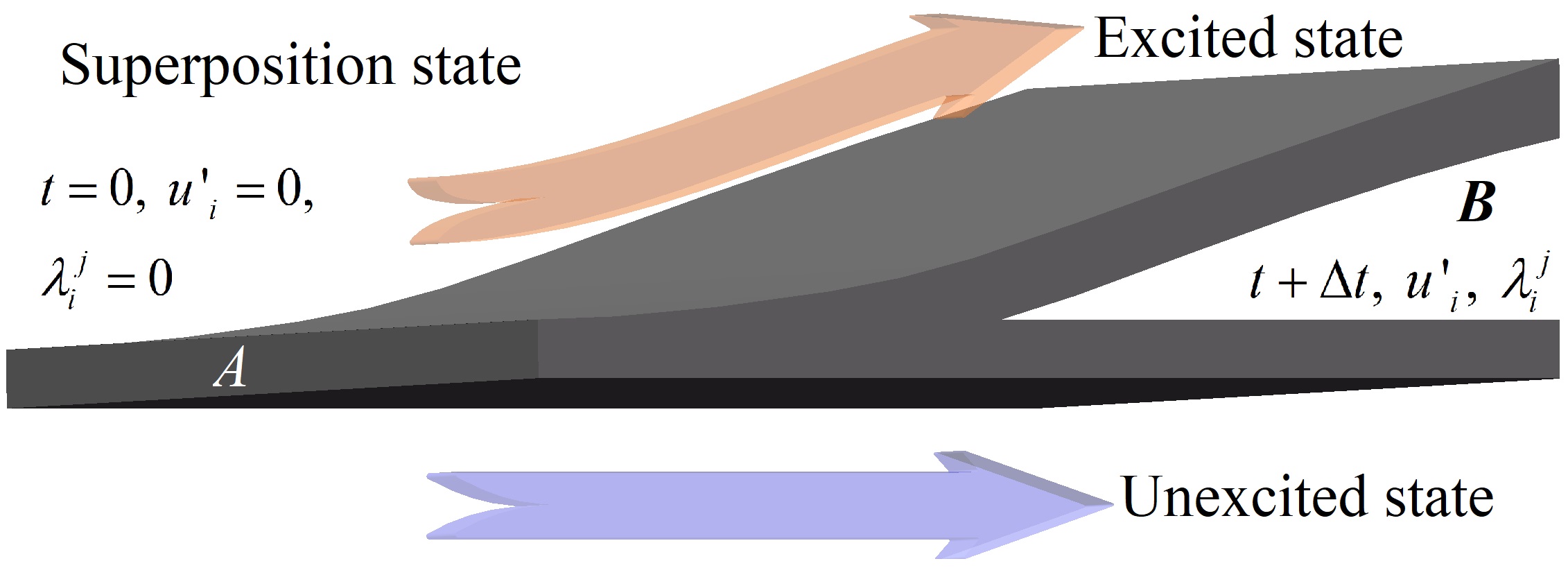}
\caption{\label{p1}Schematic diagram of superposition state ($A$ is the initial position of the superposition state; $B$ is the final position leaving the superposition state.)}
\end{figure}
As mentioned above, the fluctuation velocity and velocity strain tensor are equal to 0 (Fig.~\ref{p1}). Next, detailed mathematical analysis at position $A$ is developed. Here, it should be noted that when $\Delta t \rightarrow 0$ occurs, positions $B$ and $A$ are infinitely close to each other, and position $B$ should also satisfy the condition of position $A$ instantaneously.

Considering that the generation of the fluctuation velocity is relatively difficult to be judged, we need to define unexcited and excited states to describe various states simply.
\begin{definition}[Excited and unexcited states]
Excited state is flowfield state, when the fluctuation velocity has been generated in the flowfield; conversely, flowfield state is unexcited state.
\end{definition}
According to the sensitivity to initial conditions in the nonlinear dynamics, the flowfield will be into excited state under the necessary external physical factors. Therefore, the excitation can be considered as the cause of the change of flowfield. The reasons for flowfield excited state can be divided into two classes: active factors and passive factors. Active factors include various influences from the initial state of the fluid, such as visosity\cite{Cornelio2019}, velocity expansion\cite{Zhao2020} and so on, whereas passive factors cover the influence and interference of the geometry and physical conditions of the fluid. To date, it has not been determined which specific factors cause the excited state of the flowfield, but the basic situation of the flowfield has been discussed from the perspective of phenomenological physics in a large number of literatures, such as pressure boundary\cite{Han2021}, geometric shape\cite{Jo2021}, inflow velocity\cite{Debuysschere2021}, temperature\cite{Rauf2020} and so on. However, the temperature is generally the main factor leading to energy change, which is not the focus of this article.\par
The mathematical relationship between the excited state and the unexcited state is an important physical variable. To describe the process of an excited flowfield, we set up a new physical variable that can be used to evaluate the measurement tensor of the ratio of fluctuation velocity to average velocity.\par
\begin{definition}[Velocity strain tensor]
If the second-order tensor $\lambda ^{j}_i$ satisfies the following relationship
\begin{equation}
\label{10}
u'_i=\bar{u}_j\lambda ^{j}_i,\quad u'_i,\bar{u}_j \in \mathscr{T}_M,
\end{equation}
then $\lambda ^{j}_i$ is called the velocity strain tensor.
\end{definition}
Meanwhile, from the above definition and Eqs.~(\ref{8}) and (\ref{9}), the following theory is obtained.
\begin{theorem}[Conditions of excited state]
At the initial point of the excited state position, the velocity strain tensor $\lambda ^{j}_i$ satisfies the equation
\begin{equation}
\label{18}
\lambda^{j}_i=0,
\end{equation}
the time derivative of the velocity strain tensor is expressed as
\begin{equation}
\label{19}
\frac{\mathrm{d}\lambda ^{j}_i}{\mathrm{d}t}=\frac{\partial \lambda ^{j}_i}{\partial t}+u_k\lambda ^{j,k}_{i,\bullet}\neq 0,
\end{equation}
where $u_k$ is the velocity field of the flow.
\end{theorem}

In the unexcited state, there is no fluctuation velocity in the flowfield. The velocity of the fluid particles is only related to the average velocity of the fluid. Therefore, the fluctuation velocity of the fluid ${u}'_i$ can be regarded as zero. So, Eq.~(\ref{5}) could be simplified as
\begin{equation}
\label{11}
u_i=\bar{u}_i.
\end{equation}

By introducing Eq.~(\ref{11}) into Eqs.~(\ref{2}) and (\ref{3}) to describe the flowfield motion, the unexcited state equations can be obtained as follows. For the continuity equation,
\begin{equation}
\label{12}
\frac{\partial \rho }{\partial t}+(\rho \bar{u}_{i}^{\bullet})^{, i}=0,
\end{equation}
and for the momentum equation under the action of viscosity, the temporal and spatial variation of the flowfield can be described as 
\begin{equation}
\label{13}
\rho \frac{\partial \bar{u}_i}{\partial t}+\rho \bar{u}_k\bar{u}_{i,\bullet}^{\bullet,k}=\sigma _{ik,\bullet}^{\bullet\bullet,k}+f_i.
\end{equation}

To make the expression clearer and avoid problems in later calculations, the following theorem could be established through Eqs.~(\ref{12}) and (\ref{13}). Considering that $\delta_i^{j}$ is an unit constant tensor, its derivative is a zero tensor.

\begin{theorem}[Unexcited state] For the unexcited state, the continuity equation is
\begin{equation}
\label{14}
\frac{\partial \rho }{\partial t}+\rho^{\bullet,i}\bar{u}_j\delta_i^{j}+\rho\bar{u}_{j,\bullet}^{\bullet,i}\delta_i^{j}=0, \quad \bar{u}_j \in \mathscr{T}_M,
\end{equation}
and the momentum equation for the viscous medium is
\begin{equation}
\label{15}
\rho \frac{\partial (\bar{u}_j\delta ^{j}_i)}{\partial t}+\rho \bar{u}_{j1}\bar{u}_{j2,\bullet}^{\bullet,k}\delta_k^{j1}\delta ^{j2}_i=g_i^{(1)},
\end{equation}
where $g_i^{(1)}$ is the unexcited state external force field, there is
$$g_i^{(1)}=\sigma ^{(1)\bullet \bullet,k}_{ik,\bullet}+ f_i,$$
here $\sigma^{(1)} _{ik}$ is the general stress state tensor in unexcited state, and $f_i$ is the body force.
\end{theorem}

In the excited state, when the flowfield gets the necessary conditions, the generated fluctuation velocity is the difference of the instantaneous velocity of the flowfield and the average velocity. Therefore, to calculate the instantaneous velocity of the flowfield at this time, the fluctuation velocity is replaced by the average velocity in Eq.~(\ref{10}), and then, substituted into Eq.~(\ref{5}). To obtain a simplified expression for the flowfield velocity, the excited ratio tensor is defined.
\begin{definition}[Excited ratio tensor]
If the second-order tensor $\xi^{j}_i$ satisfies the relationship
\begin{equation}
\label{16}
\xi^{j}_i=\delta ^{j}_i+\lambda ^{j}_i
\end{equation}
then, $\xi ^{j}_i$ is the excited ratio tensor.
\end{definition}

Excited ratio tensor could be utilized to evaluate the measurement tensor to describe the ratio of average velocity to instantaneous velocity. The following conclusion is drawn.
\begin{equation}
\label{17}
u_i=\bar{u}_j\delta ^{j}_i+\bar{u}_j\lambda ^{j}_i=\bar{u}_j(\delta ^{j}_i+\lambda ^{j}_i)=\bar{u}_j\xi ^{j}_i.
\end{equation}

To better compare the phenomena without the excited flowfield, Eqs. (\ref{18}), (\ref{19}), and (\ref{17}) are combined with Eqs.~(\ref{2}) and (\ref{3}), and then, the following theorem is obtained. In order to reduce the difficulties of mathematical expression, we set two adjustment tensor coefficients $e_{ki}^{j1j2}$ and $\alpha_{i}^{j1j2}$, there are 
$$e_{ki}^{j1j2}=\xi ^{j1}_k\xi ^{j2}_i, \ \alpha_{i}^{j1j2}=\xi ^{j1}_k\xi ^{j2,k}_{i,\bullet}.$$

\begin{theorem}[Excited state] For the excited state, the excited ratio tensor $\xi^{j}_i$ satisfies the continuity equation,
\begin{equation}
\label{20}
\frac{\partial \rho }{\partial t}+\rho^{\bullet,i} \bar{u}_j\xi^{j}_i+\rho \bar{u}_{j,\bullet}^{\bullet,i}\xi ^{j}_{i}+\rho \bar{u}_j\xi ^{j,i}_{i,\bullet}=0, \quad \bar{u}_j \in \mathscr{T}_M
\end{equation}
and the momentum equation for a viscous medium
\begin{equation}
\label{21}
\rho \frac{\partial (\bar{u}_j\xi^{j}_i)}{\partial t}+\rho \bar{u}_{j1}\bar{u}_{j2,\bullet}^{\bullet,k}e_{ki}^{j1j2}+\rho \bar{u}_{j1}\bar{u}_{j2}\alpha_{i}^{j1j2}=g_i^{(2)},
\end{equation}
where $g_i^{(2)}$ is the excited state external force field, there is
$$g_i^{(2)}=\sigma ^{(2)\bullet \bullet,k}_{ik,\bullet}+f_i,$$
here $\sigma^{(2)} _{ik}$ is the general stress state tensor in excited state, and $f_i$ is the body force.
\end{theorem}

To better compare the two assumed flowfield states mathematically, the following relationship is easily derived:
\begin{equation}
\label{22}
\epsilon_{ki}^{j1j2}=e_{ki}^{j1j2}-\delta ^{j1}_k\delta ^{j2}_i=\lambda ^{j1}_k\lambda ^{j2}_i+\lambda ^{j1}_k\delta ^{j2}_i+\delta ^{j1}_k\lambda ^{j2}_i.
\end{equation}\par
The adjustment tensor coefficient $\epsilon_{ki}^{j1j2}$ indicates no definite physical meaning, just for mathematical calculation. 
The flowfield is excited by the fluid physical properties and changed under external conditions, included in Eqs.~(\ref{14}), (\ref{15}), (\ref{20}), and (\ref{21}). Therefore, the excited equation could be obtained.

\begin{theorem}[Excited equation 1] 
For the excited process, the velocity strain tensor $\lambda ^{j}_i$ satisfies the continuity equation
\begin{equation}
\label{23}
\rho^{\bullet,i}\bar{u}_j\lambda ^{j}_i+\rho \bar{u}_{j,\bullet}^{\bullet,i}\lambda ^{j}_i+\rho \bar{u}_j\lambda ^{j,i}_{i,\bullet}=0, \quad \bar{u}_j \in \mathscr{T}_M
\end{equation}
and the momentum equation for the viscous medium
\begin{equation}
\label{24}
\rho\frac{\partial\left(\bar{u}_j\lambda^{j}_i\right)}{\partial t}+\rho \bar{u}_{j1}\left(\bar{u}_{j2,\bullet}^{\bullet,k}\epsilon_{ki}^{j1j2}+\bar{u}_{j2}\xi^{j1}_k\lambda ^{j2,k}_{i,\bullet}\right)=g_i,
\end{equation}
where $g_i$ is the excited external force field, expressed as:
$$g_i=g_i^{(2)}-g_i^{(1)}=\sigma ^{(2)\bullet\bullet,k}_{ik,\bullet}-\sigma ^{(1)\bullet\bullet,k}_{ik,\bullet}.$$
\end{theorem}

At the same time, Eqs.~(\ref{14}) and (\ref{15}) under the excited condition are substituted into Eqs.~(\ref{23}) and (\ref{24}) to get the following theorem.

\begin{theorem}[Excited equation 2] If the flowfield is in the excited process under the action of physical properties, the velocity strain tensor $\lambda ^{j}_i$ satisfies the continuity equation
\begin{equation}
\label{25}
\lambda ^{j,i}_{i,\bullet}=0
\end{equation}
and the momentum equation for the viscous medium
\begin{equation}
\label{26}
\rho \bar{u}_j\frac{\partial \lambda ^{j}_i}{\partial t}+\rho \bar{u}_k\bar{u}_j \lambda^{j,k}_{i,\bullet}=g_i, \quad \bar{u} \in \mathscr{T}_M
\end{equation}
where $g_i$ is the excited external force field. 
\end{theorem}

At this point, the mathematical significance of Eq.~(\ref{25}) can be described according to Gauss' theorem, when excited state is about to take place. The excited state position is taken as the origin and an open sphere $\tau$ with radius $\varepsilon >0$ is established, the following conclusion, similar to Eq.~(\ref{25}), could be drawn:
\begin{equation}
\label{27}
\lambda ^{j,i}_{i,\bullet}=\lim_{\Delta \tau \to 0}\frac{1}{\Delta \tau }\oint_{\sigma }\lambda ^{j}_in^{i}d\sigma=0.
\end{equation}\par

\section{\label{sec:level5}Stress state analysis of excited equations}
The stress state\cite{Sui2017} at any point in the flowfield is uniquely determined by the stress vectors on the three orthogonal planes of action at that point, and each stress vector can be represented by three components again. Therefore, in this article, we assume that the fluid in the flowfield is isotropic, the relationship between stress tensor and deformation rate tensor is linear, and their function could be expressed as in Eq.~(\ref{4}).

For the unexcited state, the stress tensor of the flowfield could be expressed as: 
\begin{equation}
\label{28}
\sigma ^{(1)}_{ik}=-p_{(1)}\delta _{ik}+\psi \mu\bar{u}_{s,\bullet}^{\bullet,s}\delta _{ik}+\mu (\bar{u}_{i,k}+\bar{u}_{k,i}).
\end{equation}

Considering the consistency of the flowfield frame, the index of the variable subscripts could be transformed. So, the gradient of Eq.~(\ref{28}) could be read as:
\begin{equation}
\label{29}
\begin{split}
\sigma ^{(1)\bullet\bullet,k}_{ik,\bullet}=-p^{\bullet,k}_{(1)}\delta _{ik}+\psi \mu\bar{u}_{t,\bullet \bullet}^{\bullet,s k}\delta^{t}_{s}\delta _{ik}\\
+\mu \left(\bar{u}^{\bullet,\bullet k}_{j,k\bullet}\delta_i^j+\bar{u}_{j,i\bullet}^{\bullet,\bullet k}\delta_k^j\right).
\end{split}
\end{equation}

For the excited state, the stress tensor of the flowfield could be written as follows:
\begin{equation}
\label{30}
\begin{split}
\sigma ^{(2)}_{ik}=-p_{(2)}\delta _{ik}+ \psi\mu\bar{u}_{t,\bullet}^{\bullet,s}\xi_{s}^t\delta _{ik}+ \psi\mu\bar{u}_{t}\xi^{t,s}_{s,\bullet}\delta _{ik}\\
+\mu (\bar{u}_{j,k}\xi^j_i+\bar{u}_{j}\xi^{j,\bullet}_{i,k}+\bar{u}_{j,i}\xi^{j}_{k}+\bar{u}_{j}\xi^{j,\bullet}_{k,i}).
\end{split}
\end{equation}

The gradient of the above expression is 
\begin{equation}
\label{31}
\begin{split}
\sigma ^{(2)\bullet\bullet,k}_{ik,\bullet}=-p^{\bullet,k}_{(2)}\delta _{ik}+\psi\mu\bar{u}_{t,\bullet \bullet}^{\bullet,sk}\xi_{s}^t\delta _{ik}\\
+\psi\mu\bar{u}_{t,\bullet}^{\bullet,s}\xi_{s, \bullet}^{t,k}\delta _{ik}+\psi\mu\bar{u}_{t,\bullet}^{\bullet,k}\xi^{t,s}_{s,\bullet}\delta _{ik}+ \psi\mu\bar{u}_{t}\xi^{t,sk}_{s,\bullet \bullet}\delta _{ik}\\
+\mu\bar{u}_{j,k \bullet}^{\bullet,\bullet k}\xi^j_i+\mu\bar{u}_{j,k}\xi^{j,k}_{i,\bullet}+\mu\bar{u}_{j,\bullet}^{\bullet,k}\xi^{j,\bullet}_{i,k}+\mu\bar{u}_{j}\xi^{j,\bullet k}_{i,k \bullet}\\
+\mu\bar{u}_{j,i \bullet}^{\bullet,\bullet k}\xi^{j}_{k}+\mu\bar{u}_{j,i}\xi^{j,k}_{k,\bullet}+\mu\bar{u}_{j,\bullet}^{\bullet, k}\xi^{j,\bullet}_{k,i}+\mu\bar{u}_{j}\xi^{j,\bullet k}_{k,i \bullet}.
\end{split}
\end{equation}

Eq.~(\ref{29}) is subtracted from Eq.~(\ref{31}), and the stress difference between two states could be obtained. A combination Eqs.~(\ref{18}) and (\ref{25}), and derivative of unit tesor $\delta_i^{j}$ as one zero tensor, the final form is that 

\begin{equation}
\label{32}
\begin{split}
g_i=\sigma ^{(2)\bullet\bullet,k}_{ik,\bullet}-\sigma ^{(1)\bullet\bullet,k}_{ik,\bullet}=p_{\bullet,i}^{(1)}-p_{\bullet,i}^{(2)}+\mu\bar{u}_{j,k}\lambda^{j,k}_{i,\bullet}\\
+\psi\mu\left(\bar{u}_{j,\bullet}^{\bullet,k}\lambda_{k, i}^{j,\bullet}+ \bar{u}_{j}\lambda^{j,\bullet k}_{k,i\bullet}\right)+\mu\left(\bar{u}_{j,\bullet}^{\bullet,k}\lambda^{j,\bullet}_{i,k}+\bar{u}_{j}\lambda^{j,\bullet k}_{i,k \bullet}\right)\\
+\mu\left(\bar{u}_{j,\bullet}^{\bullet, k}\lambda^{j,\bullet}_{k,i}+\bar{u}_{j}\lambda^{j,\bullet k}_{k,i \bullet}\right).\\
\end{split}
\end{equation}

Thus, the final excitation law is obtained by synthesizing Eqs.~(\ref{26}) and (\ref{32}).
\begin{theorem}[Excitation law] To reach the excited state, the velocity strain tensor $\lambda_i^j$ satisfies the continuity condition
\begin{equation}
\label{33}
\lambda ^{j}_i=0, \quad \lambda ^{j,i}_{i,\bullet}=0
\end{equation}
and the momentum equation for the viscous medium
\begin{equation}
\begin{split}
\label{34}
\rho u_j\frac{\partial \lambda ^{j}_i}{\partial t}+\rho {u}_k{u}_j \lambda^{j,k}_{i,\bullet}=P_i+\mu{u}_{j,k}\lambda^{j,k}_{i,\bullet}\\
+\mu \left(u_j\lambda^{j,\bullet}_{i,k}+\phi u_j\lambda^{j,\bullet}_{k,i}\right)^{,k},\quad u_i \in \mathscr{T}_M,
\end{split}
\end{equation}
where $\phi=\psi+1$ is auxiliary coefficient. 
\end{theorem}

\section{\label{sec:level6}Degenerate form of excitation law}
The moment when the flowfield variations could be regarded as the superposition state. Therefore, Eqs.~(\ref{33}) and (\ref{34}) could be equivalent to the fundamental equations of flowfield instant variation. The applicable conditions are consistent with Eqs.~(\ref{2}) and (\ref{3}), which could solve the problems of separation, transition, shock wave, reattachment, and so on. Because Eqs.~(\ref{33}) and (\ref{34}) can not be solved completely at present, the equations are degenerated in order to meet the engineering needs.\par

As shown in Fig.~\ref{p1}, the degenerate form discusses the position B. Degenerate form\cite{Zhan2021} means to add some necessary degenerative conditions so that the problem can be solved further easier. Although, it breaks one certain balance via the perspective of mathematical logic, the degenerate form can help to achieve the prediction effectively and accurately. Its accuracy strongly depends on the degenerative conditions and operation method. It should be noted that for the identical problems, the form and condition of degradation are not unique. At the time of excitation, it can be considered that time is still and the shape of superposition state does not change.\par

\begin{definition}[Degenerative condition] The basic condition of flowfield degradation is
\begin{equation}
\label{35}
\frac{d x_k}{d t}=\frac{\partial x_k}{\partial t}
\end{equation}
\end{definition}

For the common physical investigations, we consider that coordinates and time are mutually independent variables\cite{Liu2021}. If space is $n$-dimensional, then $n$ coordinates are independent. In continuum mechanics, the longitudinal displacement of a deformable body can cause its lateral displacement at the same time \cite{Thota2021}. For example, a drop of water on the desktop disperses automatically, which is obviously contrary to coordinate independence. However, considering that the superposition state is a state that it is independent and the instantaneous shape is unchanged as shown in Fig.~\ref{p2}, the physical meaning of the degenerate condition  significance is to ignore these shape-change influence factors.

\begin{figure}[ht]
\includegraphics[scale=0.2,angle=0]{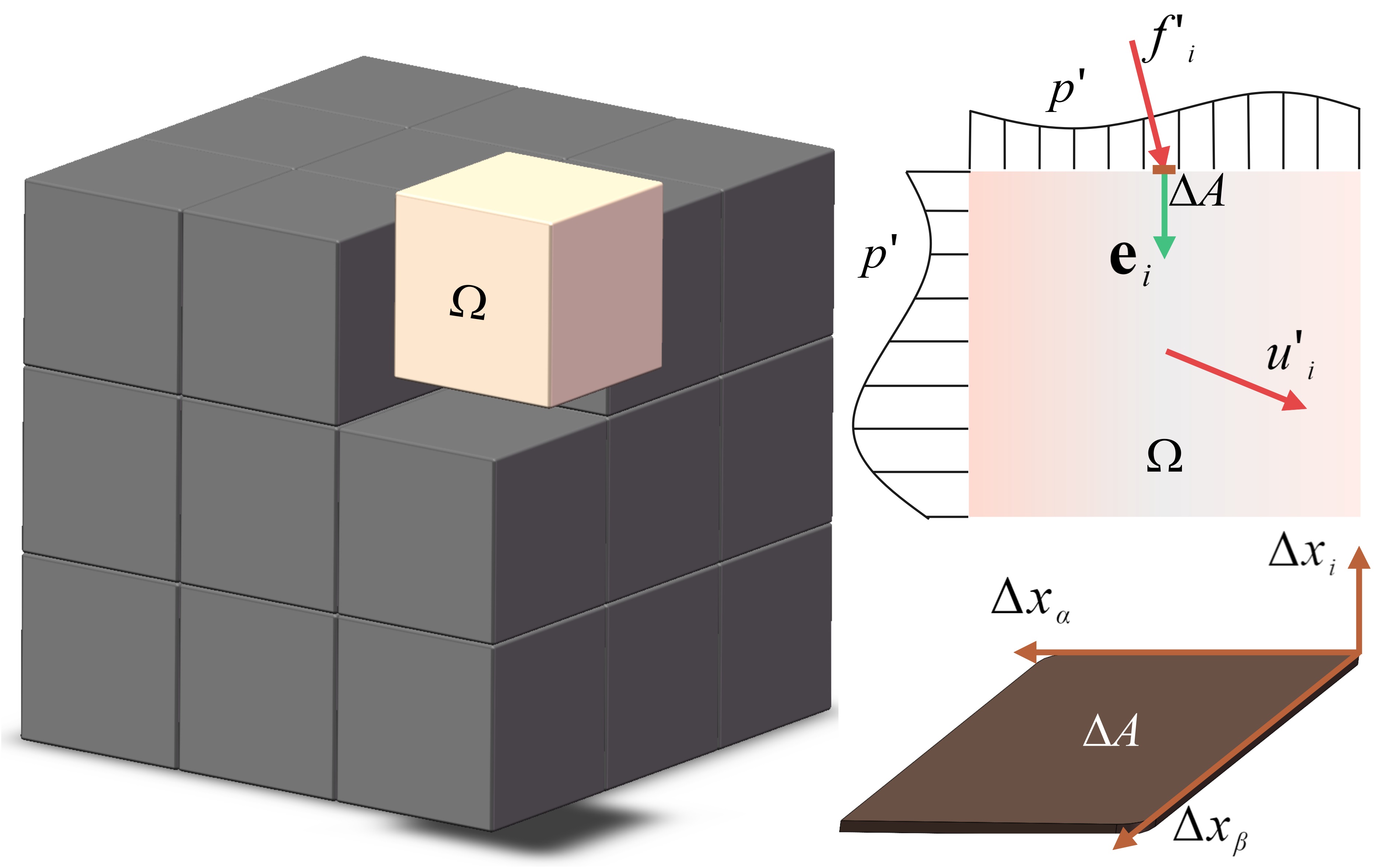}
\caption{\label{p2}Independent excited position diagram (the fluctuation velocity is generated by fluctuation force $f'_i$ or fluctuation pressure $p'$ according to the isotropic condition; $p'$ is the same in all directions)}
\end{figure}

In accordance with the above analysis methods, the fluid had previously advanced a small distance in a very short time period $\Delta t$. The velocity strain tensor develops from zero to $\lambda_i^j$. According to the difference principle, Eqs.~(\ref{33}) and (\ref{34}) could be expanded via the combination with degradation condition, shown in Eq.~(\ref{35}). While, the shape change with time, which will be deduced in detail in the discussion of precision later.\par
\begin{theorem}[Degenerate form of excitation law] If the system is in an excited state, the fluctuation velocity $u'_i \in \mathscr{T}_M$ satisfies 
\begin{equation}
\label{36}
\rho\frac{\partial u'_i}{\partial t}=\frac{P_i+P'_i}{1+\gamma},
\end{equation}
where $\gamma=1-\beta$ is the dimensionless dynamic coefficient; $P_i$ is the pressure gradient difference. Let the fluctuation velocity gradient is $C'_{ki}$, the characteristic pressure $P'_i$ owns the following form
$$P'_i=\mu \left(\phi\frac{\partial C'_{ki}}{\partial t}+\frac{\partial C'_{ik}}{\partial t}\right)\left(u_k\right)^{-1}.$$
Let velocity gradient is $C_{kj}$, the viscosity coefficient $\beta$ has the following form
$$\beta = \mu C_{kj}\left(\rho u_k u_j\right)^{-1}$$
here $\mu C_{kj}$ is the viscous stress; $\rho u_k u_j$ is the inertial stress.
\end{theorem}

\begin{proof}
If there is an excited flowfield under the action of physical properties, then the velocity strain tensor $\lambda_i^j$ satisfies
\begin{equation}\nonumber
\begin{split}
\rho u_j\frac{\partial \lambda ^{j}_i}{\partial t}+\rho {u}_k{u}_j \lambda^{j,k}_{i,\bullet}=P_i+\mu{u}_{j,k}\lambda^{j,k}_{i,\bullet}\\
+\mu \left(u_j\lambda^{j,\bullet}_{i,k}+\phi u_j\lambda^{j,\bullet}_{k,i}\right)^{,k},\quad u_i \in \mathscr{T}_M,
\end{split}
\end{equation}

The flowfield will gradually become excited state under the effect of pressure derivative differences. In a short period of time $\Delta t$, the velocity strain rate tensor changes from zero to $\lambda^{j}_i$. So, we obtain the following relationship:
\begin{equation}\nonumber
\begin{split}
\frac{\rho u'_i}{\Delta t}+\rho u_k \frac{u'_i}{\Delta x_k} =P_i +\mu \frac{\Delta u_j}{\Delta x^k \Delta x_k} u'_i\left(u_j\right)^{-1}+\\
\mu\left( \frac{u'_i}{\Delta x^k \Delta x_k}+\phi\frac{u'_k}{\Delta x^i \Delta x_k}\right),
\end{split}
\end{equation}

After a simple mathematical rearrangement, the above formula becomes
\begin{equation}\nonumber
\begin{split}
\left[\rho+ \left(\rho u_k u_j -\mu \frac{\Delta u_j}{\Delta x^k} \right)\left(\frac{\Delta x_k}{\Delta t}\right)^{-1}\left(u_j\right)^{-1}\right]u'_i\\
=P_i \Delta t+\mu\left( \frac{\Delta u'_i }{\Delta x^k} \left(\frac{\Delta x_k}{\Delta t}\right)^{-1}+\phi\frac{\Delta u'_k }{\Delta x^i}\left(\frac{\Delta x_k}{\Delta t}\right)^{-1}\right)
\end{split}
\end{equation}
as $\Delta t \rightarrow 0$, the change of each physical variable in the above formula is as follows:
$$\frac{\Delta u_j}{\Delta x^k} \rightarrow \frac{\partial u_j}{\partial x^k} = C_{jk}, \ \frac{\Delta x_k}{\Delta t} \rightarrow \frac{\partial x_k}{\partial t}=\hat{u}_k, \ \frac{\Delta u'_k}{\Delta x^i} \rightarrow \frac{\partial u'_k}{\partial x^i}=C'_{ki},$$
$$\frac{\Delta u'_i}{\Delta x^k} \rightarrow \frac{\partial u'_i}{\partial x^k}=C'_{ik},\ \frac{\Delta C'_{ik}}{\Delta t} \rightarrow \frac{\partial C'_{ik}}{\partial t},\ \frac{\Delta C'_{ki}}{\Delta t} \rightarrow \frac{\partial C'_{ki}}{\partial t}.$$

Similarly, the fluctuation velocity gradients change from zero to $C'_{ki}$ or $C'_{ik}$, we obtain
\begin{equation}\nonumber
\begin{split}
\left[1+ \left(\rho u_k u_j -\mu C_{jk}\right)\left(\rho \hat{u}_k u_j\right)^{-1}\right]\rho \frac{\Delta u'_i}{\Delta t} \\
=P_i+\mu\left(\frac{\partial C'_{ik}}{\partial t}+\phi\frac{\partial C'_{ki}}{\partial t}\right)\left(\hat{u}_k\right)^{-1}.
\end{split}
\end{equation}

Finally, combining degradation condition $u_k=\hat{u}_k$, we conclude that
$$\rho \frac{\partial u'_i}{\partial t}=\frac{1}{1+\gamma} \left(P_i+P'_i\right),$$
where $P_i$ is the pressure gradient difference; $P'_i$ is the characteristic pressure; $C_{jk}$ is the velocity gradient; $\beta$ is the dimensionless coefficient of the viscous--inertial force ratio, where
$$\gamma=1-\beta, \quad \beta = \mu C_{jk}\left(\rho u_k u_j\right)^{-1}; $$
$$P'_i=\mu \left(\phi\frac{\partial C'_{ki}}{\partial t}+\frac{\partial C'_{ik}}{\partial t}\right)\left(u_k\right)^{-1}.$$\par

This completes the proof.
\end{proof}

As the Reynolds number \cite{Chen2021} is defined as the ratio of inertial force to viscous force, $\beta$ could be regarded as the reciprocal of the ``Reynolds number at each point'' in the flowfield. This is because, in Eqs.~(\ref{2}) and (\ref{3}), the pressure term and the viscous term (diffusion term) describe the effects of pressure and viscous force on the flowfield respectively. If the velocity in the flowfield does not change significantly, then\par
\begin{equation}
\label{37}
Re=\frac{U_{0}L}{\nu}=\frac{\rho U_{0}^{2}}{\mu U_{0}/L}\approx \frac{1}{\beta}.
\end{equation}

As a state, the instantaneous velocity in the excited state could be considered to be approximately equal. In the study of excited state fiuld mechanics, $\beta$ and $\gamma$ are taken as constants, and $\beta$ is approximately equal to the reciprocal of $Re$. This is only an physics approximation  under normal circumstances, not absolutely correct. In addition, when $Re$ is particularly small, $\beta$ is very large and the reciprocal of $1+\gamma$ is especially small. Even if $\beta$ can not be regarded as a constant at this time due to the small fluctuation velocity, the flowfield also is not easy to be excited. It is meaningless to discuss the excited state fluid mechanics. Therefore, for common fluids, such as air, water, etc., in the calculation of excited state position prediction, it could be approximately considered as $\beta \approx 0$.\par
In order to reveal the important properties of excited state in Eq.~(\ref{36}) fully, we need to study the mathematical and physical properties of pressure gradient difference and characteristic pressure in detail.

\section{\label{sec:level7}Physical properties of pressure gradient difference and characteristic pressure}
According to the derivation of the equations in the previous section, the root cause of excitation is determined by external and internal, among which the external reason is the change of pressure difference around the excitation position and the internal ones are the viscous and diffusion terms of velocity. As for external causes, we now examine the changes in the pressure difference around the excited position. Studies of pressure can reflect changes in other flow parameters. According to Eq.~(\ref{36}), there is a non-negligible relationship between the stress changes in flowfield and the pressure differences.\par
However, for practical applications and experimental measurements, the pressure parameters are not always directly measured. Experimental methods\cite{Sengupta2020}, DNS\cite{Bailey2020} and LES\cite{Balin2020} are generally used, and the power spectrum semi-empirical model is widely applied in the balanced turbulence boundary of a plate, but the experimental instruments are not sufficiently accurate and the experimental preiod takes too long. In engineering applications, such as the design of aircraft \cite{Keshtegar2017} and ships \cite{Kujala2019}, a pressure difference is generated by pressure variations inside the flowfield, leading to sharp changes in flowfield velocity. Hence, the velocity gradient becomes larger, so there is a greater friction force inside the flowfield, which also indicates that stronger vortices will be generated.\par

Let us return to the idea of superposition. If there are $\textbf{n}$ possible non-interlaced states in an uncertain event, the difference and ratio of the state characteristic equation could be used to describe the causes of the induced state. The difference describes the additional effect, while the ratio describes the rate of change of the additional effect.\par
As mentioned before, when the system goes from one equilibrium state to another, the existence of a difference indicates that a new equilibrium form needs to be added in the process of breaking the equilibrium. The pressure gradient difference is in such a form.\par
\begin{equation}
\label{38}
P_i=p_{\bullet,i}^{(1)}-p_{\bullet,i}^{(2)}=-\Delta p_{\bullet,i},
\end{equation}
where $\Delta p_{\bullet,i}$ is the variation of the pressure gradient. Under external factors, to achieve equilibrium after excitation, we must provide the equilibrium form of Eq.~(\ref{36}) on the basis of the original equilibrium. This change is caused by the temporal or spatial distribution of the pressure gradient or characteristic pressure. \par
The divergence operator is commonly used to describe the convergence degree of spatial distribution of a vector field, and the Laplace operator distribution of pressure represents the effect of some external influences on the system. In a Minkowski space\cite{Jimenez2021}, the Laplace operator is equivalent to the d'Alembert operator\cite{Ciaglia2020}. Let $\zeta$ evaluate the spatial variation rate of the pressure gradient difference, and $\Delta x_i \in \Omega$ represents the scale of spatial change. Then, according to Eq.~(\ref{38}), we obtain\par

\begin{equation}
\label{39}
\zeta = \lim\limits_{\Delta x_i\to 0}\frac{\Delta p_{\bullet,i}}{\Delta x_i}=\mathscr{L}p,
\end{equation}
where $\mathscr{L}$ is the Laplace operator. Using the same idea, we divide the two sides of Eq.~(\ref{36}) by $\Delta x_i$. Since the fluctuation velocity $u'_i$ and characteristic pressure $P'_i$ vary from 0, the following theorem could be got.

\begin{theorem}[Finally degenerate form of excitation law] If the system is excited, the fluctuation velocity $u'_i \in \mathscr{T}_M$ satisfies
\begin{equation}
\label{40}
\frac{\partial \theta'}{\partial t}=\frac{1}{1+\gamma}\left(\frac{P'^{\bullet,i}_{i,\bullet}}{\rho}-\frac{\mathscr{L}p}{\rho}\right),
\end{equation}
where $\nu$ is kinematic viscosity coefficient, $\theta'=u'^{\bullet,j}_{j,\bullet}$ is fluctuation velocity expansion, and $p$ is pressure.
\end{theorem}

Now, we have the last difficulty, which is the mathematical and physical properties of characteristic pressure. In fact, as shown in Fig.~\ref{p2}, the characteristic pressure is caused by the viscous stress of fluctuation velocity. Its divergence could be expressed, that
if $u_k=0$, then $P'^{\bullet,i}_{i,\bullet}=0$;
if $u_k \neq 0$, then 
$$P'^{\bullet,i}_{i,\bullet}=\frac{\partial}{\partial x_i}\left(\mu\left(\phi\frac{\partial C'_{ki}}{\partial t}+\frac{\partial C'_{ik}}{\partial t}\right)\left({u_k}\right)^{-1}\right).$$ 

The velocity strain tensor develops from zero to $\lambda_i^j$. By the difference principle, the time partial derivative of fluctuation velocity could be expanded and we obtain the degenerate form of the excitation law. With $\Delta t \rightarrow 0$ under degradation condition, we could get
$$\frac{\mu}{n+1}\left(\frac{\Delta C'_{ik}}{\Delta t}+n\frac{\Delta C'_{ik}}{\Delta t}\right)=\frac{\mu}{n+1}\left(\frac{\Delta u'_i}{\Delta x^k \Delta t}+\frac{\Delta u'_i\Delta x_j }{\Delta x^k \Delta t \Delta x_j}\right)$$
$$=\frac{\mu}{n+1}\frac{1}{\Delta x^k}\left(\frac{\partial u'_i}{\partial t}+\frac{\partial x_j }{\partial t}\frac{\partial u'_i}{ \partial x_j}\right)=\frac{\mu}{n+1}\frac{a'_i}{\Delta x^k}.$$

In the same way, its similar conjugate tensor
$$\phi\mu\frac{\Delta C'_{ki}}{\Delta t}=\frac{\phi\mu}{n+1}\frac{a'_k}{\Delta x^i}.$$

Therefore, considering that the fluctuation acceleration $a'_i$ is caused by the fluctuation force physically, the fluctuation inertia body force $f'_i$ is introduced, and then
$$a'_i=\frac{f'_i}{\rho}, \quad \frac{u_k}{\Delta x_k}=\frac{u_i}{\Delta x_i}, \quad \Delta x^i\Delta x_i u_k= \Delta x^i u_i\Delta x_k=\Delta x_k\Delta x^i u_i.$$\par
If $x_i$ is regarded as the normal vector of a plane, which is perpendicular to the fluctuation force $f'_i$. Considered   the geometric meaning of the cross product, the normal derivative could be  regarded as the plane derivative, and the plane derivative of the fluctuation force is the isotropic fluctuation pressure, whose value changes from 0, as shown in Fig.~\ref{p2}.
$$\frac{\Delta f'_i}{ \Delta x_i}=\frac{\Delta f'_i}{|\Delta x_\alpha \times \Delta x_\beta| \mathbf{e}_i}=\frac{\Delta f'_i \mathbf{e}^i}{\Delta A}=np',$$
as $\Delta t \rightarrow 0$, combining Eq.~(\ref{2}), the above formulas can be transformed into the following one
$$\frac{\mu}{n+1}\frac{a'_i}{\Delta x^k \Delta x_i}\left(u_k\right)^{-1}=\frac{\mu}{n+1}\frac{\Delta \left(f'_i/ \rho\right)}{\Delta x^k \Delta x_i}\left(u_k\right)^{-1}$$
$$=\frac{\mu}{\rho}\frac{1}{n+1}\frac{\Delta f'_i}{\Delta x^k \Delta x_i}\left(u_k\right)^{-1}+\frac{\mu}{n+1}\frac{f'_i\Delta\left(1/\rho\right) }{\Delta x^k \Delta x_i}\left(u_k\right)^{-1}$$
$$=\frac{\mu}{\rho}\frac{1}{n+1}\frac{\Delta f'_i}{\Delta x^k \Delta x_i}\left(u_k\right)^{-1}-\frac{\mu}{\rho}\frac{1}{n+1}\frac{\Delta f'_i}{\Delta x_i}\frac{1}{\rho}\frac{\Delta \rho }{ \Delta x^k}\left(u_k\right)^{-1}.$$

Consider the following equation
$$\frac{\Delta \rho }{ \Delta x^k}\left(u_k\right)^{-1}=\frac{1}{n+1}\left(\frac{\Delta \rho}{\Delta x^k}+n\frac{\Delta \rho}{\Delta x^k}\right)\left(u_k\right)^{-1}$$
$$=\frac{1}{n+1}\left(\frac{\Delta x^k}{\Delta t}\right)^{-1}\left(\frac{\Delta \rho}{\Delta t}+\frac{\Delta x^j}{\Delta t}\frac{\Delta \rho}{\Delta x^j}\right)\left(u_k\right)^{-1}=\frac{\eta}{n+1}\frac{d \rho}{d t}.$$

Therefore, the above equations could be expressed as:
$$\frac{\mu}{n+1}\frac{a'_i}{\Delta x^k \Delta x_i}\left(u_k\right)^{-1}=\frac{n\nu}{n+1}\frac{\Delta p'}{\Delta x^k}\left(u_k\right)^{-1}-\frac{n\nu \eta p'}{\left(n+1\right)^2}\frac{d \rho}{\rho dt}$$
$$=\frac{n\nu}{n+1} \left(p'_{\bullet,k}\left(u_k\right)^{-1}+\frac{\theta\eta p'}{n+1}\right),$$
and its similar conjugate tensor is
$$\frac{\phi\mu}{n+1}\frac{a'_k}{\Delta x^i \Delta x_i}\left(u_k\right)^{-1}=\frac{\phi\mu}{n+1}\frac{a'_k}{\Delta x^i \Delta x_k}\left(u_i\right)^{-1}=\frac{n\phi\nu p'_{\bullet,i}}{n+1}\left(u_i\right)^{-1}$$
$$-\frac{n\phi\nu\eta p'}{\left(n+1\right)^2}\frac{d \rho}{\rho d t}=\frac{n\phi\nu}{n+1} \left(p'_{\bullet,i}\left(u_i\right)^{-1}+\frac{\theta\eta p'}{n+1}\right);$$
where $\eta=\delta_i^j\left(u_i u^j\right)^{-1}$ is coefficient and $\theta=u_{j,\bullet}^{\bullet,j}$ is velocity expansion.

In addition, by the same principle, the following relationships can be obtained
$$\frac{\mu\Delta C'_{ik}}{\Delta t}=\frac{\mu \Delta a'_i}{\left(n+1\right)\Delta x^k}=\frac{\nu f'_{i,k}}{n+1}-\frac{\nu f'_{i}}{\left(n+1\right)^2}\frac{1}{\rho}\frac{d \rho}{dt}\left(u^k\right)^{-1}$$
$$=\left[\frac{n \nu p'u_i}{n+1}+\frac{\theta \nu f'_{i}}{\left(n+1\right)^2}\right]\left(u^k\right)^{-1},$$
$$\frac{\phi\mu\Delta C'_{ki}}{\Delta t}=\frac{\phi\mu \Delta a'_k}{\left(n+1\right)\Delta x^i}=\frac{\phi\nu f'_{k,i}}{n+1}-\frac{\phi \nu f'_{k}}{\left(n+1\right)^2}\frac{1}{\rho}\frac{d \rho}{dt}\left(u^i\right)^{-1}$$
$$=\left[\frac{n \phi \nu p'u_k}{n+1}+\frac{\phi \theta \nu f'_{k}}{\left(n+1\right)^2}\right]\left(u^i\right)^{-1}.$$

Finally, we could obtain the following relationship.
\begin{theorem}[Characteristic pressure divergence]
Let $v^i=\left(u_i\right)^{-1}$, the characteristic pressure divergence could be expressed as:
\begin{equation}
\label{41}
\begin{split}
P'^{\bullet,i}_{i,\bullet}=\frac{n\nu}{n+1}\left[\left(1+\phi\right)p'_{\bullet,k}v^k-\eta p'\left(u_i v_k+\phi u_k v_i\right)C^{ki}\right]\\
+\frac{\nu\eta\theta}{\left(n+1\right)^2}\left[\left(1+\phi\right)n p'-\left(f'_i v_k+\phi f'_k v_i\right) C^{ki}\right],
\end{split}
\end{equation}
where $n$ is the spatial dimension of the problem and $\nu$ is kinematic viscosity coefficient.
\end{theorem}

The fluctuation pressure $p'_i$ is an important reference variable. Some experiments have been conducted to describe the role of this variable. The pressure distribution of 51 monitored points on the upper and down surfaces of airfoil model (Type: S809) is measured by PSI electronic scanning pressure method, and 27 points by dynamic pressure sensor. 

\begin{figure}[ht]
\includegraphics[scale=0.28,angle=0]{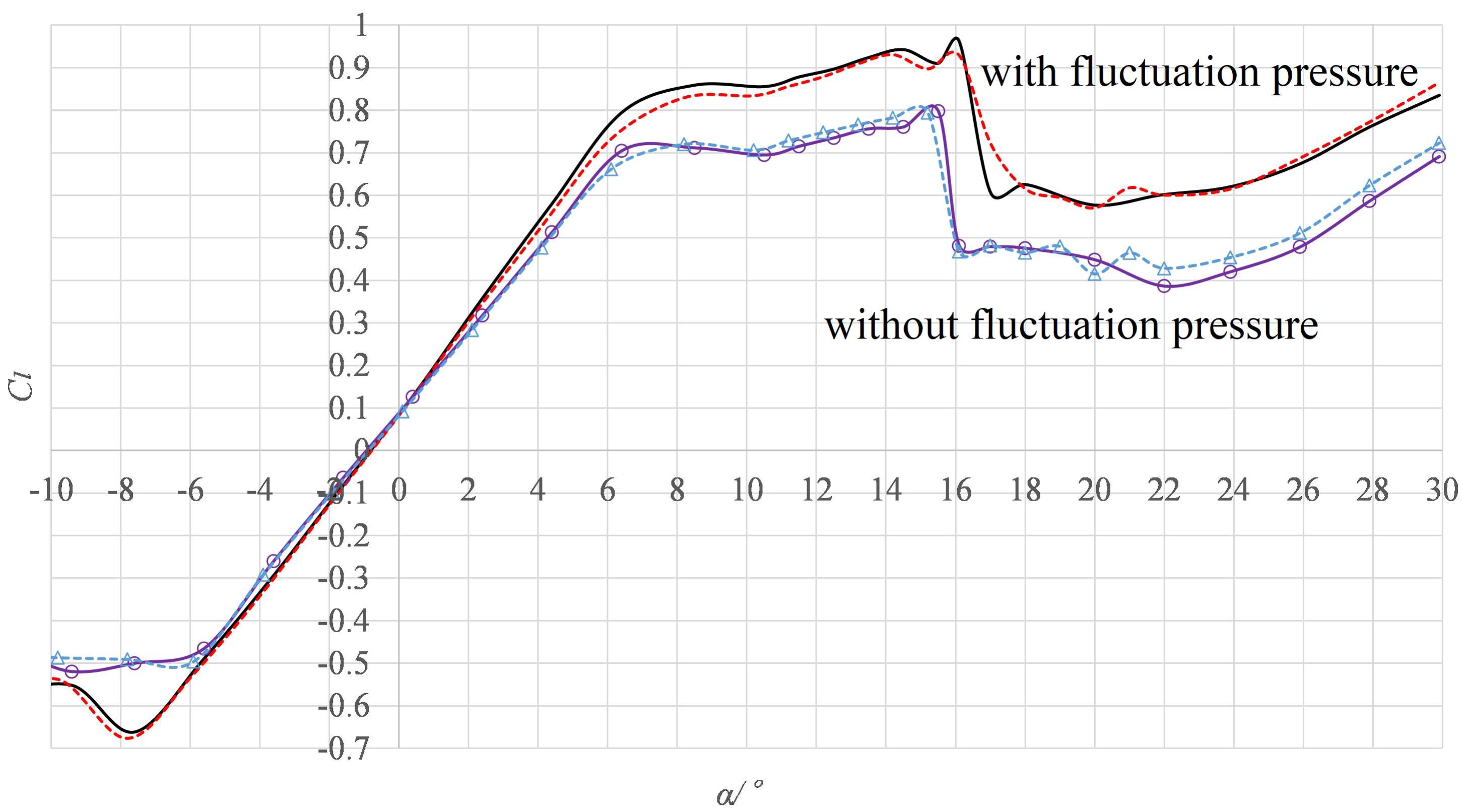}
\caption{\label{p3}Lift coefficient curve of airfoil S809 (The black curve is measured by FL-11 dynamic pressure hole at free transition; the red curve is dynamic pressure hole at fixed transition; the purple curve is static pressure hole at free transition; the green curve is static pressure hole at fixed transition)}
\end{figure}

Obviously, from Fig.~\ref{p3}, the effects of fluctuation pressure near the separation zone are more significant than those in the non-separation zone. Theoretically, transition could occur under any angle of attack. Generally speaking, the effects of fluctuation pressure on transition could be ignored, while the effects on separation are not ignored.

We can draw the following conclusions in excited state fluid mechanics:\par
\begin{enumerate}[(1)]
\item When the pressure does not change, viscosity is necessary for the flowfield excitation. According to  Eq.~(\ref{41}), the occurrence of excited state is closely related to the kinematic viscosity coefficient. The larger the kinematic viscosity coefficient is, the easier the flowfield enters the excited state. For example, because the air kinematic viscosity coefficient is larger than that of water, the flowfield is easier to be excited.
\item When the velocity is great, the flowfield is difficult to be excited. So, the excitation often occurs near the wall. For example, a large number of experimental studies \cite{Croci2019, Melius2018, Chandra2019, Miro2019} have proved that transition, separation, shock wave, reattachment, and other special flow patterns are easy to occur near the wall.
\item The stronger the expansion of the flowfield is, the easier the flowfield is excited when the flowfield velocity does not change. For the incompressible flow, effects of the fluctuation pressure are considered to be close to zero, which could be ignored in Eq.~(\ref{41}).
\item For the compressible flow, let $ \alpha $ be the sound velocity, and the flowfield is approximately one-dimensional  when the velocity and expansion of the flowfield change. Then
$$\eta\theta=\frac{1}{\alpha Ma^2}\frac{dMa}{dx},$$
\begin{enumerate}[1)]
\item If the flow is subsonic, the value of 
$$\left[\frac{dMa}{dx}\right]<<\left[Ma^2\right],$$
where $\left[a\right]$ expresses the influence degree of physical variable $a$. So, it is not easy to be excited;
\item If the flow is supersonic and $Ma$ changes suddenly , the value of 
$$\left[\frac{dMa}{dx}\right]>>\left[Ma^2\right].$$
So, it is easy to induce excitation, such as separation, shock wave, and reattachment. 
\end{enumerate}
In short, for supersonic flow, the characteristic pressure plays an important role, not be ignored.
\end{enumerate} 

Apparently, these conclusions are consistent with many existing ones. But they are not enough to explain the reliability and accuracy of quantitative calculation of degradation condition and the corresponding form. From the perspective of prediction accuracy, it should be
$$\alpha=1-\frac{\partial x_i}{\partial t}\left(u_i\right)^{-1}=1-\frac{1}{m},$$
where the velocity vector could be expressed as
$$\frac{d x_i}{d t}=\frac{\partial x_i}{\partial t}+\left(\frac{\partial x_{j_{1}}}{\partial t}+\left(\frac{\partial x_{j_{2}}}{\partial t}+...\right)\frac{\partial x_{j_{1}}}{\partial x_{j_{2}}}\right)\frac{\partial x_{i}}{\partial x_{j_{1}}}=m\frac{\partial x_i}{\partial t}.$$

Theoretically, when $m=1$, the prediction error under the degradation condition is 0; when $m \rightarrow \infty$, there is a great uncertainty in the infinitesimal and infinital domains. Therefore, as a theoretical study, it can be considered that the degradation condition is equivalent to the basic theory of excited state fluid mechanics. As an application fundamental theory in engineering practice, it is suggested to solve the optimization model or degradation condition through a large number of engineering practices in order to overcome the error caused by a large number of uncertainties.

Thus far, we have established the excited state theory of a flowfield through mathematical analysis. We now need to describe and analyze specific problems according to the definitions of separation and transition.

\section{\label{sec:level7}Theorems of transition and separation}
Transition and separation are two special cases in the excited state. One is the generation of an irregular fluctuation velocity, and the other one is the generation of reflux near the boundary layer. Therefore, we use the excitation law to derive the laws of transition and separation.\par
Transition represents irregular fluctuation velocity, and the initial value will be very small. Separation is to form a local fluctuation velocity distribution of opposite magnitude, such as Newton's internal friction velocity distribution, which could counteract the original flow velocity or produce backflow. Therefore, the effects of viscosity on separation will be much stronger than that of transition.\par

Usually, for transition, a large number of experimental investigations have proved in the process of laminar flow to turbulence, they are $u'_i \rightarrow 0$, $f'_i \rightarrow 0$, $p' \rightarrow 0$, and $p'_{\bullet,k} \rightarrow 0$. Therefore, it could be considered that
$$P'^{\bullet,i}_{i,\bullet} \rightarrow 0,$$
this is only the normal situation. In extreme cases, Eq.~(\ref{41}) shall prevail.

For separation, experimental results show that the velocity is in the state of excited state completion. So, the local velocity gradient is $C^{ki} \rightarrow 0$. Therefore, it could be considered that
$$P'^{\bullet,i}_{i,\bullet} \rightarrow \frac{n\nu \left(1+\phi\right)}{n+1}\left(p'_{\bullet,k}v^k+\frac{\theta\eta p'}{n+1}\right),$$
again, this is only the normal situation. Under extreme cases, Eq.~(\ref{41}) shall prevail as well, especially for the  incompressible flow.

\begin{theorem}[Transition]
Let $U\left(x_h,\epsilon\right)$ be a neighborhood of point $x_h \in \Omega$. The laminar flow with pressure $p \in \mathscr{T}_{M}$ runs through the neighborhood, and excites a transition by various disturbances, then 
\begin{equation}
\label{44}
\mathscr{L}p \left(x_h\right) \neq 0, \quad \mathscr{L}p \left(x_k\right) = 0, \quad \forall x_k \in U\left(x_h,\epsilon\right),
\end{equation}
where $\mathscr{L}$ is the Laplace operator. All points $x_j$ satisfying the above equation constitute a set $K \subset \Omega$, which denotes the transition position.
\end{theorem}

\begin{proof} By definition 2, $\forall x_k \in U\left(x_h,\epsilon\right)$ and $\exists x_s \in U\left(x_h,\epsilon\right)$ such that
$$u'_j\left(x_k\right) \neq u'_j\left(x_s\right),\ u'_j\left(x_h\right)=0.$$

Thus, no matter how $x_k$ changes, $u'_j$ is not equal to zero. In this case, the field $u'_j$ could not be regarded as an isolated system, and there are sources and sinks everywhere in the field. Therefore, the divergence field of the fluctuation velocity is not zero, except the point $x_h$.
$$u'^{\bullet,j}_{j,\bullet}\left({x_k}\right) \neq 0, \ u'^{\bullet,j}_{j,\bullet}\left(x_h\right)=0, \ \forall x_k \in U\left(x_h,\epsilon\right).$$

From the above discussions, $u'_j$ develops from zero in a very short time. So, we could write
$$\left(\frac{du'_{j,\bullet}}{dt}\right)^{\bullet,j}\left({x_k}\right) \neq 0, \ \left(\frac{du'_{j,\bullet}}{dt}\right)^{\bullet,j}\left({x_h}\right)=0, \ \forall x_k \in U\left(x_h,\epsilon\right).$$

$p'$ is ignored. A combination of Eq.(\ref{40}) and $u_i\left(x_s\right)=u_i\left(x_k\right)$, $\forall x_k \in U\left(x_h,\epsilon\right)$, we have that
$$\mathscr{L}p \left(x_h\right) \neq 0, \quad \mathscr{L}p \left(x_k\right) = 0.$$

In this way, we prove the necessity of the theorem. Similarly, we prove its sufficiency as well.
\end{proof}

\begin{theorem}[Separation]
Consider a scalar pressure field $p$, and a velocity vector field $u_i \in \mathscr{T}_M$. Expansion $\theta$ presents velocity divergence $u^{\bullet,i}_{i,\bullet}$. If the separation is excited near the wall, there exists a set of solutions $x_j$, satisfying the following equation:
$$\frac{\partial \theta}{\partial t}=\frac{1}{1+\gamma}\left(\frac{\mathscr{L}p}{\rho}-\frac{n\nu \left(1+\phi\right)}{n+1}\left(\frac{p'_{\bullet,k}}{\rho}v^k+ \frac{p'}{\rho}\frac{\theta\eta}{n+1}\right)\right),$$ 
then, all points $ x_j \in K$ satisfying the above equation constitute a set $K$ which denotes the separation position.
\end{theorem}

\begin{proof}
By definition 3, $U\left(x_h,\epsilon\right)$ is a neighborhood of point $x_h \in \Omega$. All $ x_k \in U \left(x_h,\epsilon\right)$ satisfy
$$u_i\left(x_k\right)=-{u'}_j\left(x_k\right), \ {u'}_j\left(x_h\right) = 0.$$

According to Eq.~(\ref{41}), we obtain
$$\rho \frac{\partial u_{i,\bullet}^{\bullet,i}}{\partial t}=\frac{\mathscr{L}p-P'^{\bullet,i}_{i,\bullet}}{1+\gamma}.$$

After a simple mathematical rearrangement, we obtain the final result, considering that the velocity gradient could be ignored during separation, $\forall x_j \in K$
$$\frac{\partial \theta}{\partial t}=\frac{1}{1+\gamma}\left(\frac{\mathscr{L}p}{\rho}-\frac{n\nu \left(1+\phi\right)}{\rho\left(n+1\right)}\left(p'_{\bullet,k}v^k+\frac{\theta\eta p'}{n+1}\right)\right).$$ 

In this way, we have proved the necessity of the theorem. Similarly, we could prove its sufficiency too.
\end{proof}

In addition, the wall condition is an important factor in continuous stable fluctuation flowfield, such as separation. The streamlines cannot pass through a solid. Therefore, for an continuous stable excited flowfield, the fluctuation acceleration direction is along the outer normal vector of the wall. According to Eq.~(\ref{39}), the pressure gradient difference direction can only follow the direction of the outer normal vector on the wall, and we could get the following theorem.\par

\begin{theorem}[Wall condition]
The flowfield near the wall is excited. Then, let scalar pressure is $p$, and character pressure is $P'_i$, which should satisfy the following relation near the key point:\par
if $n^i>0$, then
\begin{equation}\nonumber
\mathscr{L}p - P'^{\bullet,i}_{i,\bullet}<0;
\end{equation}

conversely, then
\begin{equation}
\label{43}
\mathscr{L}p - P'^{\bullet,i}_{i,\bullet}\geqslant0,
\end{equation}
where $n^i$ is the outer normal vector of the solid boundary.
\end{theorem}

For continuous stable fluctuation flowfield, another important property is that the wall condition is satisfied continuously. If it is not met with the wall condition continuously, it is easy to be affected by the wall and make the flowfield return to the previous state. The transition does not belong to the fluctuation  continuous stable flowfield; so, the transition does not need to consider the wall condition. However, if the transition occurs near the wall, the transition should be regarded as continuous stable flow. Note that when the wall condition is employed,  pressure term in Eq.~(\ref{43}) is not the pressure gradient, but the gradient difference, which is  different in nature and cannot be confused.\par
In excited state fiuld mehanics, some of fluctuation velocity divergence means the source, and some means the sink. This is not important for transition, but crucial for separation. The reason is that separation needs to produce necessary conditions to cancel the original velocity field.\par

\begin{theorem}[Divergence condition]
If the excited flowfield needs to cancel the original velocity field, let scalar pressure is $p$, and character pressure is $P'_i$, which should satisfy the following relation near the key point:
\begin{equation}
\label{44}
\mathscr{L}p - P'^{\bullet,i}_{i,\bullet}<0,
\end{equation}
where $n^i$ is the outer normal vector of the solid boundary.
\end{theorem}

In summary, we have established the fundamental theories of separation and transition under general circumstances, but they are still not conducive to engineering applications, especially for active or passive flow control technology. The theories have become more common in the aerospace field in recent years. According to the excitation law, we now define a new physical variable by considering the divergence of a fluctuation acceleration field. This variable is called the  flowfield excitation intensity.

\begin{definition}[Flowfield excitation intensity] 
For a pressure field $p$, if the system is excited, there is a scalar physical variable $\Pi$ that satisfies
\begin{equation}
\label{45}
\Pi =\frac{1}{1+\gamma}\left(\frac{\mathscr{L}p}{\rho}-\frac{P'^{\bullet,i}_{i,\bullet}}{\rho}\right),
\end{equation}
where $\mathscr{L}$ is the Laplace operator; $\rho$ is the fluid density; $\Pi$ is the flowfield excitation intensity.
\end{definition}

The goal of active or passive flow control is to eliminate the excitation intensity at the initial point of the excited state. In addition, a turbulent flow can be transformed into a laminar flow using this physical variable. As long as the excitation intensity of each point in the flowfield is known, it is possible to offset the excitation intensity.

\section{\label{sec:level8}Experimental validation and verification of theories and discussions on correlation analysis methods}
Introducing this relationship into Eq.~(\ref{44}), we find that the transition is only related to the second derivative of pressure.
For transition, the following relationship should be satisfied,
$$\frac{d^2 C_p}{dx^2}=0;$$
for separation, the following one should be met,\par
$$\frac{d^2 C_p}{dx^2}=\chi,$$
where $C_p$ is the pressure coefficient, $x$ is the separation or transition positions along the chord of the airfoil, $\epsilon \rightarrow 0$  is infinitesimal, and $n_x$ is the component of the outer normal vector of the upper or lower surfaces of the airfoil on the $x$ axis. The coefficient $\chi$ is expressed
$$\chi=\frac{n}{n+1}\left(1+\phi\right)\frac{\nu C'_p}{u_{max}}, \quad \frac{u}{u_{\infty}}\approx h\left(\frac{y}{\delta}\right),$$
where  $u_{\infty}=30m/s$ is the inflow velocity, $n=1$ is the spatial dimension, and the mapping $h$ is exprssed as the power function space, such as quadratic function space, cubic function space, etc. If the product of fluctuation pressure and dynamic viscosity is insignificant, then $\chi \rightarrow 0$. 
Considering that the superposition state is one state which is in the separation of instantaneous its dynamic equilibrium. Ignoring the air gravity, and according to d'Alembert's principle $p'\approx p_{\infty}$, the following equation holds.
$$\frac{C'_p}{C_p} \approx \frac{u^2_{\infty}}{u^2_{max}}.$$\par
Due to the wall and divergence conditions, the separation position should be met the following conditions
$$n^i>0, \quad \frac{d^2 C_p}{dx^2}\left(x-\epsilon\right)<0;$$
the transition position should meet the following conditions.
If $n^i>0$, then
\begin{equation}\nonumber
\frac{d^2 C_p}{dx^2}\left(x-\epsilon\right)<0;
\end{equation}
Conversely, then
\begin{equation}\nonumber
\frac{d^2 C_p}{dx^2}\left(x-\epsilon\right)\geqslant0,
\end{equation}
where $n^i$ is the outer normal vector of the solid boundary. For the airfiol NACA2412, when the point is at the position which is less than the maximum bending position, the normal vector is regarded as negative $n^i<0$; contrarily, the normal vector is regarded as positive $n^i>0$.\par
Experiments are the final verification of all physical theory results. Therefore, we will employ several simple one-dimensional flow experiments to reveal the applicability of the separation and transition theory. Because it is the one-dimensional flow, the sum of all physical variables has one degree of freedom. In the experiment, an airfoil NACA2412  with an inflow velocity of 30 m/s was used. To analyze the transition and separation, let its angle of attack be $12^{\circ} $ and $-2^{\circ} $. A diagram of the experimental device is shown in Fig.~\ref{p4}.

In the first case, the angle of attack is $-2^{\circ}$. Product of the fluctuation pressure and dynamic viscosity is insignificant; so, it can be considered $\delta \rightarrow 0$. The pressure curve and the corresponding second derivative curve are shown in the Figs.~{\ref{p5} and \ref{p6}}. The experimental data are shown in Fig~.\ref{p7}. The initial transition position is $x=0.405$ and the full transition position is $x=0.502$. Region C is an error due to the viscosity of the fixed wall, which is equivalent to an infinite straight wing without wall for the airfoil used in the prediction. This error should be ignored.
\begin{figure}[ht]
\includegraphics[scale=0.39,angle=270]{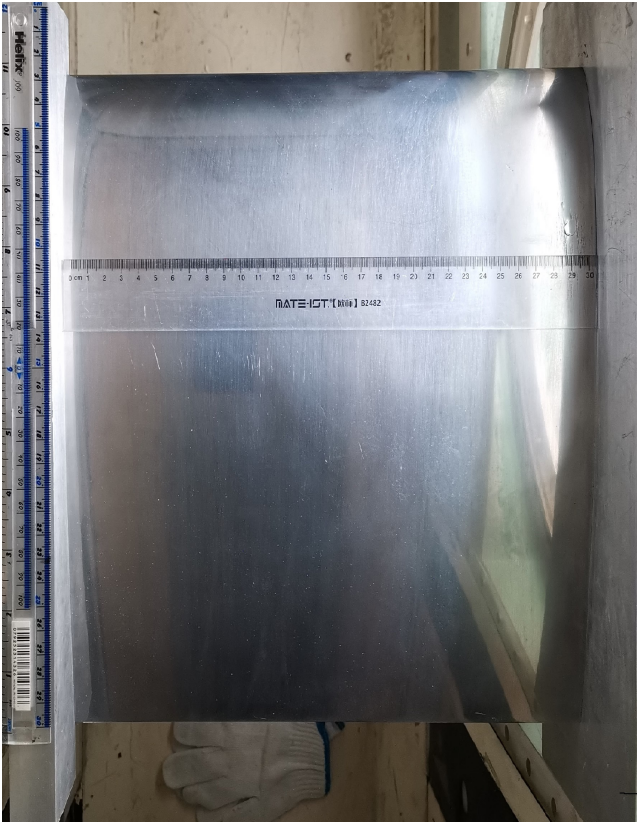}
\caption{\label{p4}Installation and fixation of airfoil NACA2412.}
\end{figure}\par

\begin{figure}[ht]
\includegraphics[scale=0.32,angle=0]{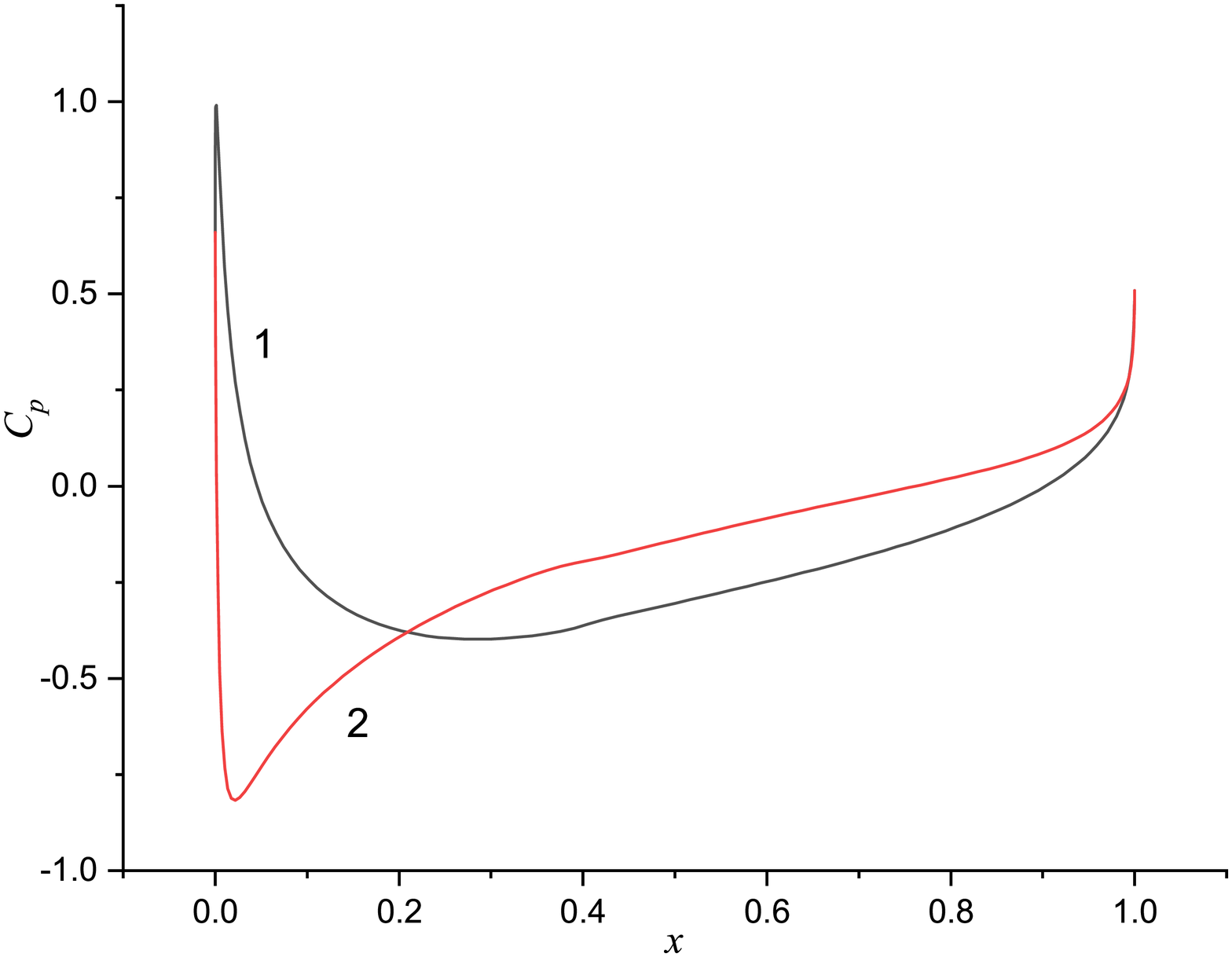}
\caption{\label{p5}Pressure curve of NACA2412 with attack angle of $-2^{\circ}$ (1 is the upper surface of the airfoil and 2 is the lower surface.)}
\end{figure}

\begin{figure}[ht]
\includegraphics[scale=0.32,angle=0]{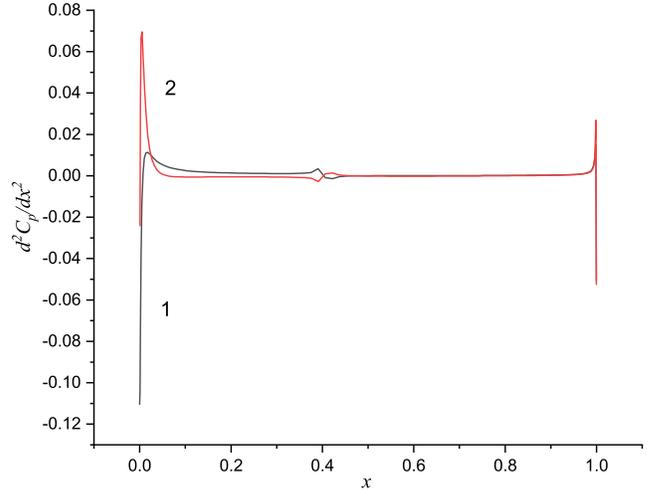}
\caption{\label{p6}Second derivative curve of pressure of NACA2412 with attack angle of $-2^{\circ}$. (1 is the upper surface of the airfoil and 2 is the lower surface.)}
\end{figure}

Observe curve 2 (lower surface) in Fig.~\ref{p6}. When the normal vector component $n^i$ is greater than $0$, the second derivative of the pressure curve of lower surface is greater than $0$ before the second derivative value equals $0$; when it is less than $0$, it has the contrary variations. There is no separation or transition occurence on the lower surface. Similairly, observe curve 1 (upper surface) in Fig.~\ref{p6}. When the normal vector component $n^i$ is less than $0$, there is no point to let the second derivative equal to $0$ and before it, the second derivative is greater than 0. There is no separation or transition on the lower surface; when it is greater than $0$, the second derivative of the pressure curve of lower surface is greater than $0$ before the second derivative value equals $0$. Here, there are both separation and transition effects. But when the separation and transition occur at the same time, the irregularity of the chaotic fluctuation velosity excited by the transition can cause the fluid to remain on the boundary by the definition 3. From the perspective of traditional fuild mechanics, the separation does not occur and there is only transition here, seen in Fig.~\ref{p7} that the predicted value is $x=0.403$ with transition and $x=0.508$ with full transition.

\begin{figure}[ht]
\includegraphics[scale=0.33,angle=0]{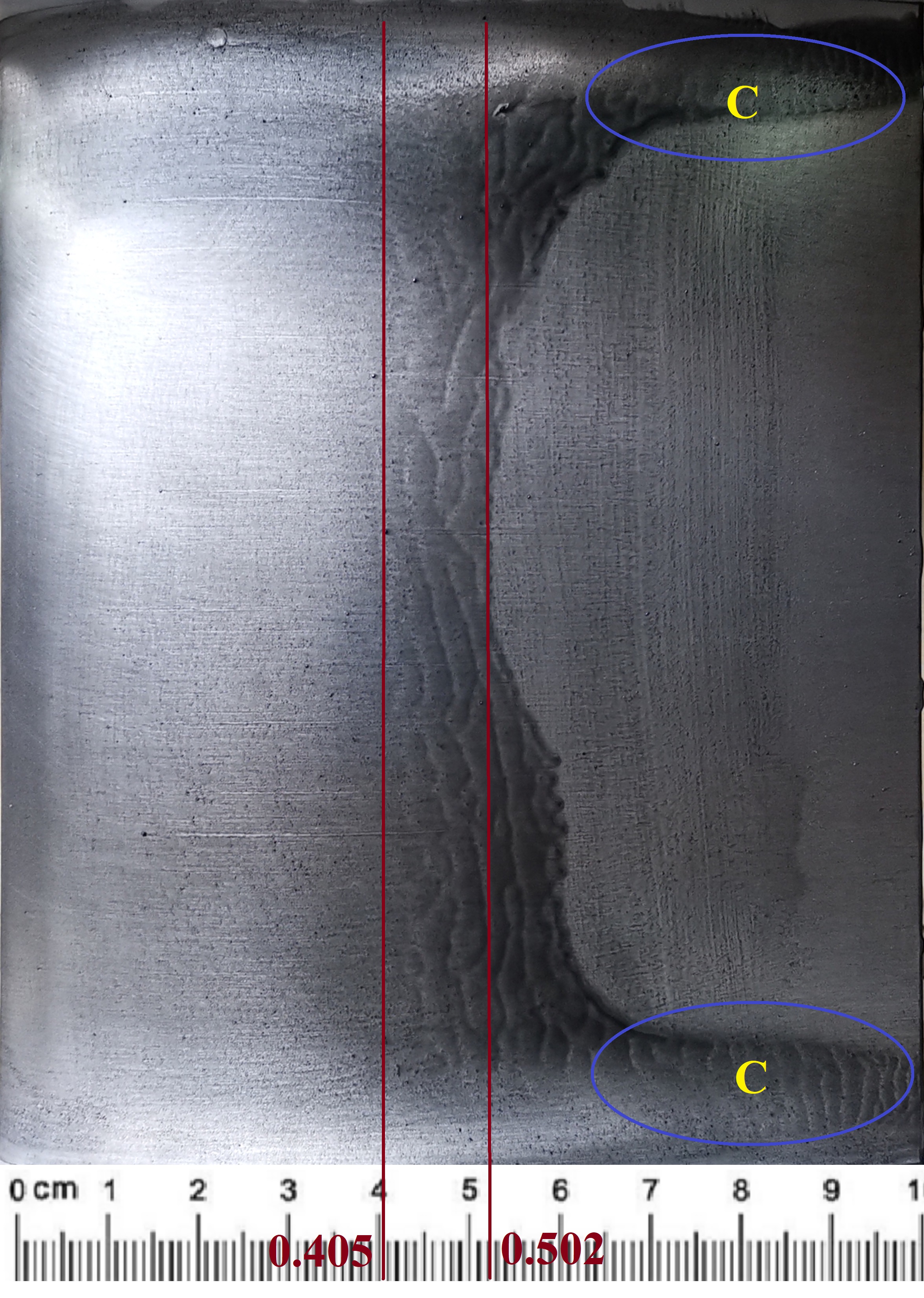}
\caption{\label{p7}Schematic diagram of the oil flow experiment result of  airfoil NACA2412 with  angle of attack  $-2^{\circ}$.}
\end{figure}

\begin{figure}[ht]
\includegraphics[scale=0.32,angle=0]{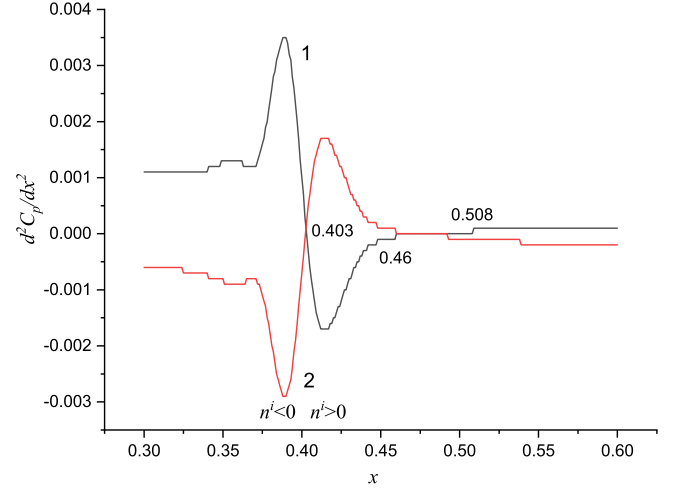}
\caption{\label{p8}The second derivative curve of pressure of airfoil NACA2412 with  angle of attack $-2^{\circ}$$, x \in \left[0.3, 0.6\right]$ (1 is the upper surface of the airfoil and 2 is the lower surface.)}
\end{figure}

Take the values from Fig.~\ref{p8}, the prediction accuracy could be calculated as
$$\Delta_1=\frac{|0.403-0.405|}{0.405}=0.49\%;$$
$$\Delta_2=\frac{|0.508-0.502|}{0.502}=1.19\%.$$

For the second case, the angle of attack is $12^{\circ}$. The corresponding product is not insignificant, constructed function from $\chi$.

When the above equation is intersected with the pressure quadratic curve, the result solved presents the separation point. The pressure curve and the second derivative curve are shown in the Figs.~{\ref{p9} and \ref{p10}}. The experimental data are shown in Fig~.\ref{p11}. The transition position is $x=0.0071$ and the separation position is $x=0.755$. Region E shows the separation bubble $x \in \left(0.5, 2.5\right) \cup \left(2.5, 6\right)$; region D is an error due to the viscosity of the fixed wall, which is equivalent to an infinite straight wing without wall for the airfoil used in the prediction. This error should be ignored as well.

\begin{figure}[ht]
\includegraphics[scale=0.32,angle=0]{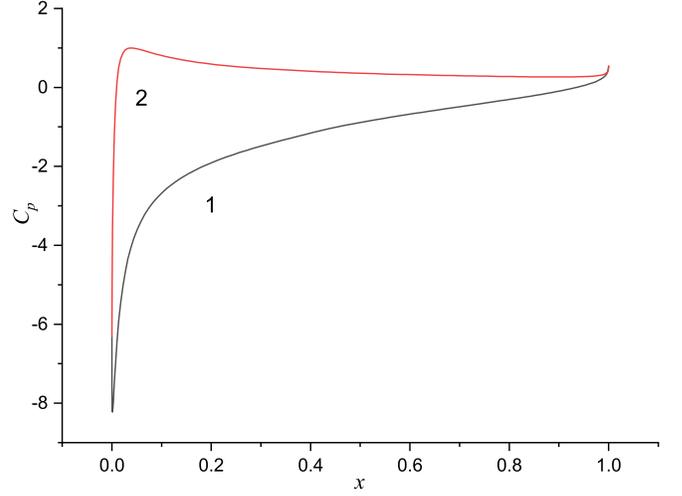}
\caption{\label{p9}Pressure curve of NACA2412 with attack angle of $12^{\circ}$, where 1 is the upper surface of the airfoil and 2 is the lower surface of the airfoil.}
\end{figure}

\begin{figure}[ht]
\includegraphics[scale=0.32,angle=0]{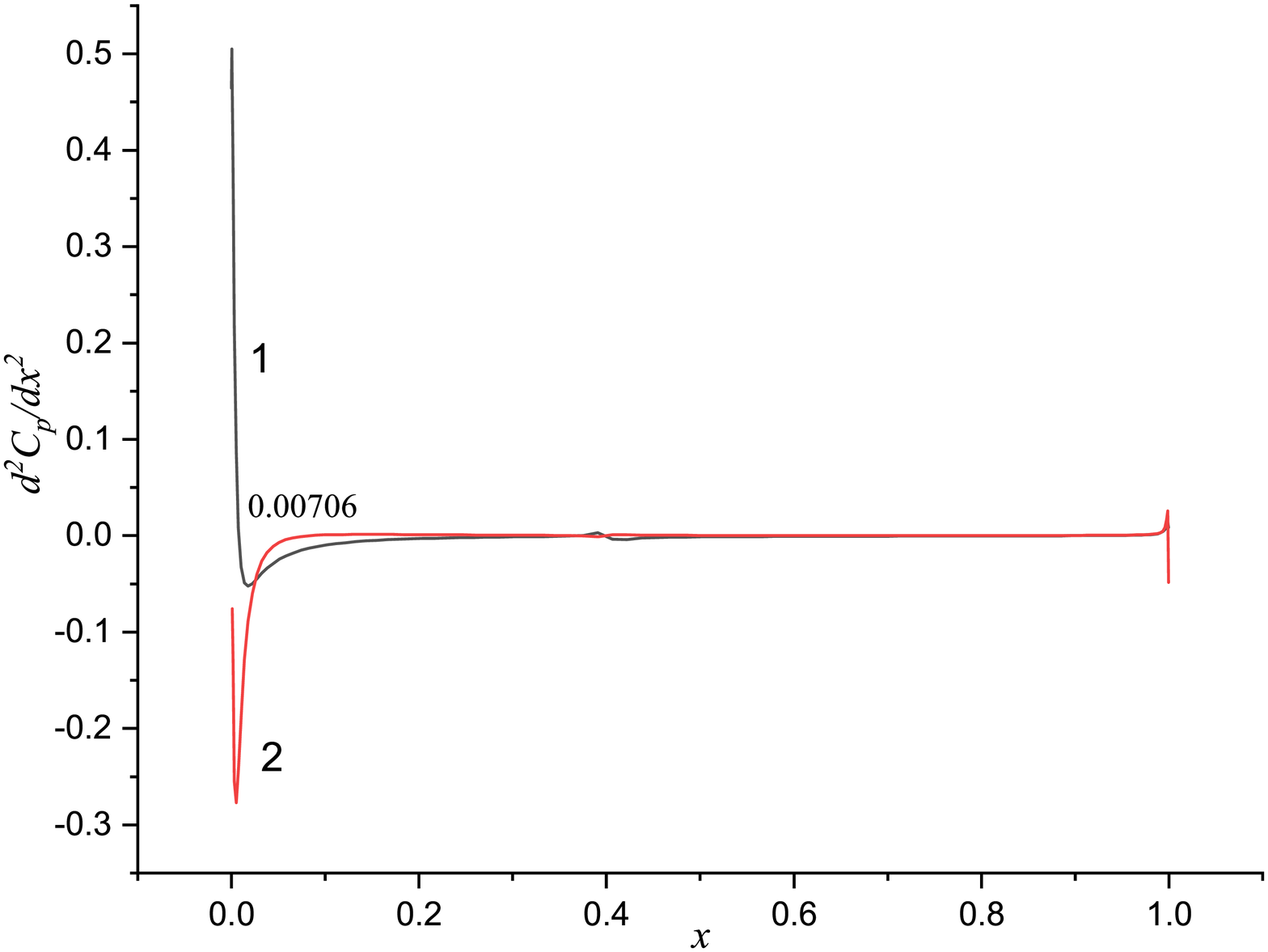}
\caption{\label{p10}Second derivative curve of pressure of NACA2412 with attack angle of $12^{\circ}$, where 1 is the upper surface of the airfoil and 2 is the lower surface of the airfoil.}
\end{figure}

Observe curve 2 (lower surface) in Fig.~\ref{p10}. When the normal vector component $n^i$ is less than $0$, the second derivative of the pressure curve of lower surface is less than $0$ before the second derivative value equals $0$; when it is greater than $0$, the second derivative of pressure with in most areas of curve 2 is greater than $0$ except the point near $x=0.4$, where the transition can be occured. Observe curve 1 (upper surface) in Fig.~\ref{p10}. When the normal vector component $n^i$ is less than $0$, there is  one point $x=0.00709$, where the second derivative is equal to $0$ and before it, the second derivative is greater than $0$. Transition occurs on the upper surface. When it is greater than $0$, the intersection point of pressure second derivative curve and function $\chi\left(x\right)$ is $x=0.762$ (Fig.~\ref{12}) and the second derivative of the pressure on upper surface is less than $0$. The separation occurs at this point. Also, it could be seen from Fig.~\ref{12} that there are two intersections between curve $\chi\left(x\right)$ and second derivative curve of pressure with $x = 0.2879$ and $x=0.4005$. Since a closed set is formed here, separation bubbles are easy to be formed here. The prediction accuracy can be calculated as
$$\Delta_1=\frac{|0.0071-0.00709|}{0.71}=1.4\%;$$
$$\Delta_2=\frac{|0.755-0.762|}{0.755}=0.927\%.$$ 

\begin{figure}[ht]
\includegraphics[scale=0.47,angle=0]{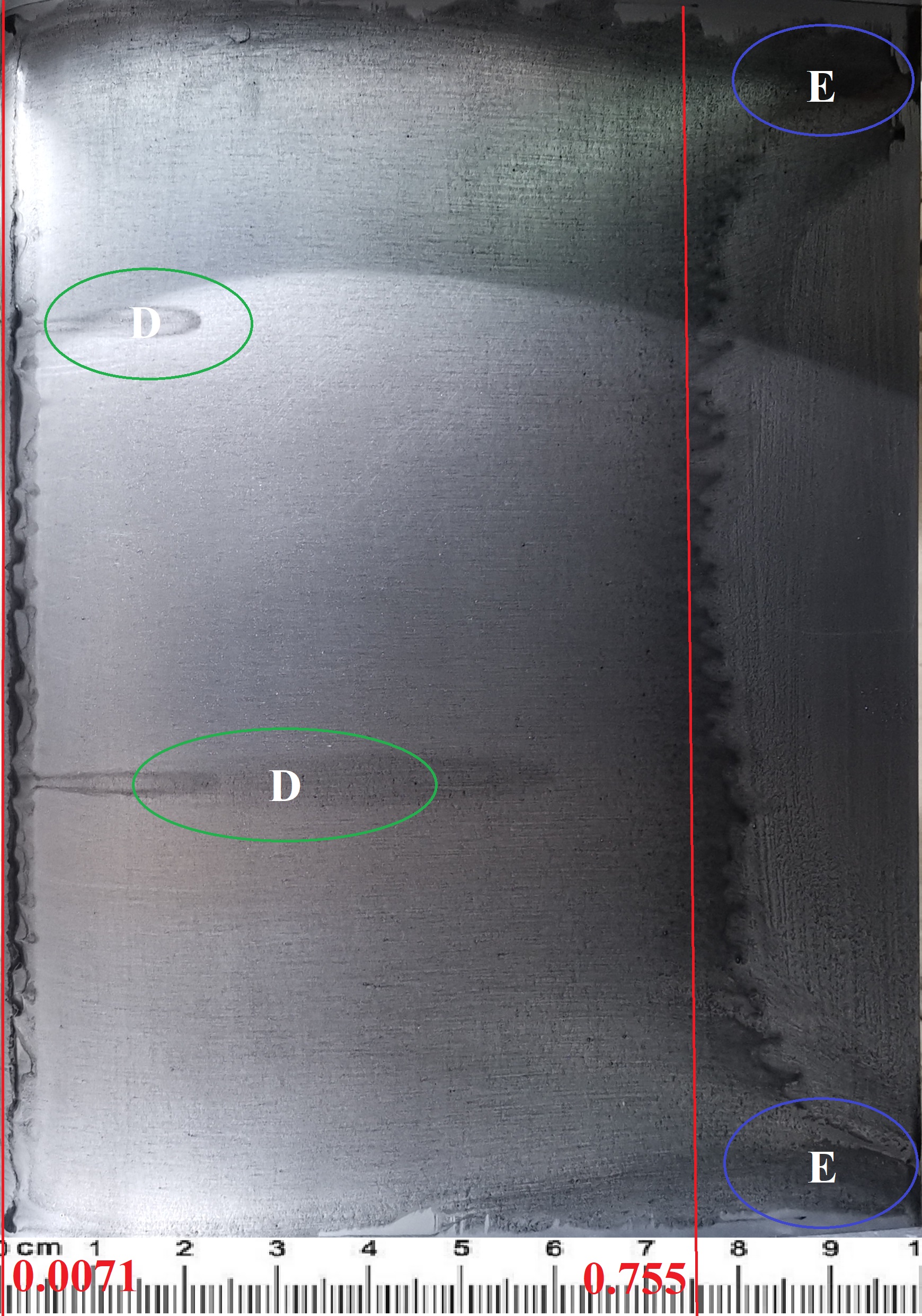}
\caption{\label{p11}Schematic diagram of the oil flow experiment result of airfiol NACA2412 with  angle of attack $12^{\circ}$.}
\end{figure}

\begin{figure}[ht]
\includegraphics[scale=0.31,angle=0]{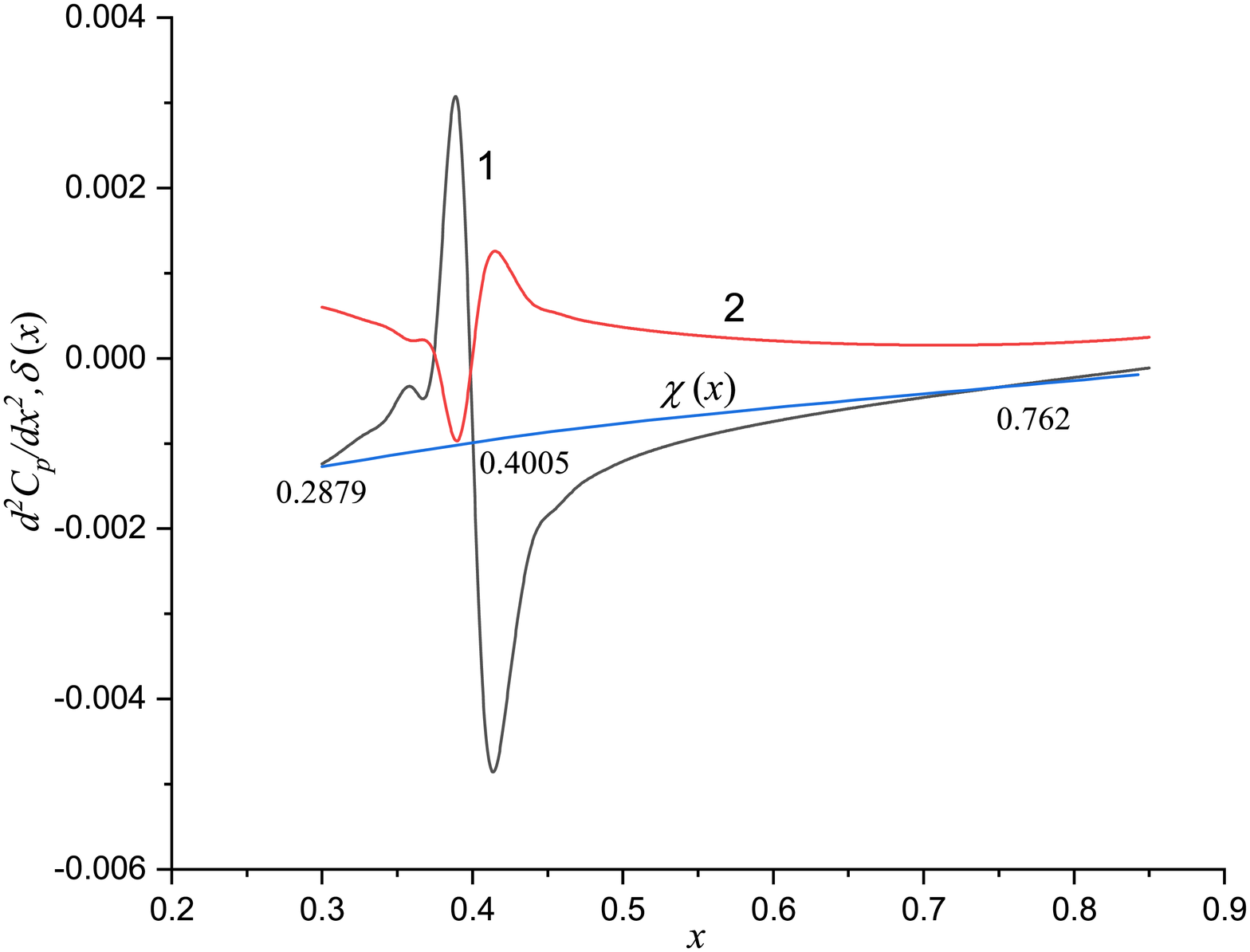}
\caption{\label{p12}The second derivative curve of pressure of airfoil NACA2412 with attack angle of $12^{\circ}$, $x \in \left[0.3, 0.85\right]$ (1 is the upper surface of the airfoil and 2 is the lower surface.)}
\end{figure}

Because the measurement equipment employed in the experiment has a certain measurement error, the prediction accuracy of the theory should be within $2\%$ under conservative estimation. The main reason for the decline of prediction accuracy of separated bubbles is the different research dimensions. The theoretical prediction study is one-dimensional, the experimental study is three-dimensional, and the separation bubbles have obvious three-dimensional characteristics. Therefore, if we can carry out three-dimensional experimental analysis and investigations on more accurate pressure parameters, more accurate verification accuracy will be obtained. Generally speaking, the correctness of the theory, at least in the low speed region and subsonic region, has been verified.

\section{\label{sec:level9}Conclusions}

In engineering, the determination of the separation and transition positions  affects the accuracy of test results, the feasibility of design results, and the applicability of test standards directly or indirectly.  These problems also limit in the development and improvement of active or passive flow control technologies. In order to solve them effectively, in this work, via the employment of rational mechanics and modern applied mathematics methods, with the aid of superposition state idea of quantum mechanics, the method of solving the variations of flowfield by using the excited state methodology of fluid mechanics is proposed, and the fundamental tensor equation of excited state is obtained. In order to further establish the basic relationship between separation and transition, the concept of degenerate form is proposed; moreover, by discretizing time and space, the excited degenerate equation is obtained, and the basic theories of transition and separation are proposed. Compared with the experimental results, the accuracy of the conservative estimation is less than $2\%$, which meets the engineering requirements.\par
For the theories of fluid mechanics , this work obtains the characteristic pressure. When the pressure gradient does not change, the characteristic pressure is the direct cause of the variations of flowfield. For characteristic pressure, from the perspective of fluid physics, the fluid viscosity is the fundamental cause of the changes of flowfield. In supersonic flow, velocity expansion is the main factor to change the flowfield characteristics. In addition, the fluctuation force caused by fluctuation pressure is one new physical variable, and its experimental methods and essential significances could not be clarified at present. This variable will also cause the flowfield enters the excited state under the effects of viscosity. Because the excited state covers a wide range of flowfield variation properties, including separation, transition, shock wave, reattachment, and cavitation, it is impossible to determine whether this physical variable contains other implied characteristics or not.\par
Theoretically speaking, although the prediction error of the degradation condition is $0$, there is a great uncertainty in the infinitesimal and infinital domains. Therefore, as a theoretical study, it could be considered that the degradation condition is equivalent to the fundamental theories of excited state fluid mechanics. But, as an application  in engineering practice, it is suggested to solve the optimization model or degradation condition through a large number of engineering practice by adopt the method in order to overcome the error caused by a large number of uncertainties. Especially, when the fiuld mechanics problems under various extreme conditions are studied, it should pay more attentions to analyze the applicable conditions of the theories.\par
This is a new research direction, in the future, it can help other engineering disciplines to carry out a lot of optimization investigations, including pressure distribution under complex geometric boundary, new flow control methods, innovative design methods of fluid machinery, and rational mechanical analysis methods of various general and extreme flowfield changes under excited state theories.\par

\section*{AUTHORS' CONTRIBUTIONS}
P.Y. conceived the present research idea, and proposed and developed the mathematical model and analysis method; J.X. designed the experiment to verify the mathematical model and analysis method; K.X. analyzed and verified the mathematical model; P.Y., J.X., and K. X. drafted the manuscript; M.L. and F.J. constructed the experimental setup and performed the experiments; Y.L. and D.P. checked and proofread the manuscript; P.Y. and J.X. approved the final manuscript for submission.

\begin{acknowledgments}
The authors wish to express their thanks to all the editors for their hand work and assistance. Thanks are also to the University of Electronic Science and Technology of China and the China Aerodynamics Research and Development Center for providing excellent conditions for this project. At the same time, we are grateful for the financial support of the Sichuan Province Expert Service Center (Grant No. M162019LXHGKJHD18) and the Fundamental Research Funds for the Central Universities (Grant No. A03019023801181) in the course of the project development.
\end{acknowledgments}

\section*{DATA AVAILABILITY}
The data that support the findings of this study are available from the corresponding author upon reasonable request.

\nocite{*}
\bibliography{main}

\end{document}